\documentclass[review,onefignum,onetabnum]{siamart171218}

\usepackage{amsmath,amssymb,bm}
\usepackage{graphicx}
\usepackage{subcaption}
\usepackage{tikz}
\usepackage{mathtools}
\usepackage{siunitx}
\usepackage{algorithm}
\usepackage{algorithmic}
\usepackage{enumerate}
\usepackage{multirow}
\usepackage{float}

\usepackage{lineno}

\nolinenumbers

\usepackage{lipsum}
\usepackage{amsfonts}
\usepackage{graphicx}
\usepackage{epstopdf}
\usepackage{algorithmic}
\ifpdf
  \DeclareGraphicsExtensions{.eps,.pdf,.png,.jpg}
\else
  \DeclareGraphicsExtensions{.eps}
\fi

\newsiamremark{remark}{Remark}

\newsiamremark{hypothesis}{Hypothesis}
\crefname{hypothesis}{Hypothesis}{Hypotheses}
\newsiamthm{assumption}{Assumption}
\newsiamthm{claim}{Claim}

\newsiamremark{example}{Example}

\headers{Flux-Preserving Adaptive Finite State Projection}{A. Dendukuri}

\title{Flux-Preserving Adaptive Finite State Projection for Multiscale Stochastic Reaction Networks
\thanks{Submitted to the editors \today.}}

\author{Aditya Dendukuri\thanks{Department of Computer Science, University of California, Santa Barbara
  (\email{aditya\_dendukuri@ucsb.edu}, \email{petzold@engineering.ucsb.edu}).}
\and Shivkumar Chandrasekaran\thanks{Department of Electrical and Computer Engineering, University of California, Santa Barbara 
  (\email{shiv@ece.ucsb.edu}).}
\and Linda Petzold\footnotemark[2]}

\usepackage{amsopn}
\DeclareMathOperator{\diag}{diag}

\makeatletter
\newcommand*{\addFileDependency}[1]{
  \typeout{(#1)}
  \@addtofilelist{#1}
  \IfFileExists{#1}{}{\typeout{No file #1.}}
}
\makeatother

\makeatletter
\AtBeginDocument{%
  \def\refstepcounter@optarg[#1]#2{%
    \cref@old@refstepcounter{#2}%
    \cref@constructprefix{#2}{\cref@result}%
    \@ifundefined{cref@#1@alias}%
      {\def\@tempa{#1}}%
      {\def\@tempa{\csname cref@#1@alias\endcsname}}%
    \protected@edef\cref@currentlabel{%
      [\@tempa][\arabic{#2}][\cref@result]%
      \csname p@#2\endcsname\csname the#2\endcsname}}%
}
\makeatother

\ifpdf
\hypersetup{
  pdftitle={Flux-Preserving Adaptive Finite State Projection},
  pdfauthor={A. Dendukuri S. Chandrasekaran and L. Petzold}
}
\fi

\begin{document}
\maketitle

\begin{abstract}
The Finite State Projection (FSP) method approximates the Chemical Master Equation (CME) by restricting the dynamics to a finite subset of the (typically infinite) state space, enabling direct numerical solution with computable error bounds. Adaptive variants update this subset in time, but multiscale systems with widely separated reaction rates remain challenging, as low-probability bottleneck states can carry essential probability flux and the dynamics alternate between fast transients and slowly evolving stiff regimes. We propose a flux-based adaptive FSP method that uses probability flux to drive both state-space pruning and time-step selection. The pruning rule protects low-probability states with large outgoing flux, preserving connectivity in bottleneck systems, while the time-step rule adapts to the instantaneous total flux to handle rate constants spanning several orders of magnitude. Numerical experiments on stiff, oscillatory, and bottleneck reaction networks show that the method maintains accuracy while using substantially smaller state spaces.
\end{abstract}

\begin{keywords}
  chemical master equation, finite state projection, adaptive methods, stiff stochastic systems, multiscale modeling, flux-based pruning
\end{keywords}

\begin{AMS}
  60H35, 65C40, 60J27  
\end{AMS}

\section{Introduction}
The Chemical Master Equation (CME) governs the evolution of probability distributions for well-mixed stochastic reaction networks and is the canonical description of intrinsic noise in biochemical systems \cite{Gillespie1977, Wilkinson2006-so}. 
Formally, the CME is an infinite system of linear ODEs posed on a countably infinite lattice of molecular population states. 
Because the reachable state space typically grows combinatorially with copy numbers, the CME cannot be solved directly without an adaptive or carefully structured truncation scheme.

An important observation for practical biochemical systems is that their probability distributions are naturally sparse: most of the theoretically infinite state space carries negligible probability. The Finite State Projection (FSP) method \cite{Munsky2006} exploits this sparsity by restricting the CME to a finite subset of states that carry most of the probability mass, yielding a finite-dimensional system that can be solved exactly together with \emph{a priori} bounds on the truncation error. Several optimizations and variants have since been developed~\cite{Kazeev2014, Vo2017, Dinh_2020,Sunkara2012, Jahnke2010, Burrage2006AKF, Sidje2015}. We briefly review some of these variants in Section~\ref{sec:related-work} and direct the reader to \cite{Dinh2016} for a detailed literature review. 

In this paper we focus on stiff biochemical reaction networks \cite{Gillespie2001, CAO20083472, Wilkinson2006-so}, where reaction rates span many orders of magnitude and the dynamics must be tracked over long time horizons. Stiff stochastic systems create several numerical complications and require adaptive methods that adjust simulation parameters locally in time to match the instantaneous demands of the dynamics. We introduce a flux-based adaptive FSP method that addresses both spatial and temporal stiffness. The central idea is to use a quantitative measure of boundary activity, given by the rate at which the probability distribution flows, to adaptively control both state space truncation and time stepping and to obtain computable error bounds without stationary information. This allows the method to retain states with very low probability that are nevertheless essential for preserving connectivity of the underlying Markov process. The time step selection uses the system flux as an activity indicator and automatically adjusts step sizes across multiple time scales. 

\subsection{Related Work}
\label{sec:related-work}
We review prior work on adaptive FSP methods and state space truncation strategies for stochastic chemical kinetics. The original FSP method \cite{Munsky2006} introduced fixed truncation with rigorous error bounds based on the probability flux leaving the truncated domain. Munsky and Khammash \cite{MUNSKY2007818} extended this framework to time-stepping variants that repeatedly expand the state space along reachable directions, providing adaptive control of the truncation error over time.

Peles et al.~\cite{peles} combined FSP with time scale separation techniques, using eigenvalue-based decompositions to separate slow and fast subsystems and aggregating rarely visited states into sink states with \emph{a priori} error bounds. Related aggregation techniques for stiff Markov chains were developed by Bobbio and Trivedi \cite{aggregation}, who compute cumulative measures efficiently by lumping rapidly evolving subsets of states. This approach achieves substantial model reduction but requires identifying and separating fast and slow dynamics \emph{a priori}, which is difficult when time scales are not well separated. Zhang, Watson, and Cao \cite{zhang_cao} introduced an adaptive aggregation method for the CME that groups micro-states into macro-states using information gathered from Monte Carlo simulation, thereby reducing the effective dimensionality of the system despite the lack of explicit \emph{a priori} error estimates.

MacNamara et al.\cite{shev} developed hybrid FSP--leap (FSP--tau-leaping) schemes using Krylov subspace methods, showing that CME-based solvers can efficiently detect equilibrium, compute moments, and generate approximate sample paths for bistable systems such as the genetic toggle switch. Kuntz et al.~\cite{kuntz2010} introduced the exit-time finite state projection (ETFSP) approach, a truncation-based extension of FSP that yields lower bounds on exit distributions and occupation measures, together with computable total-variation error bounds that decrease monotonically and converge as the truncation is enlarged. For reaction networks with conserved quantities, the \textit{slack reactant} method \cite{slack} exploits stoichiometric constraints to reduce the effective dimensionality of the state space, pruning states that violate conservation laws.

The rest of the paper is organized as follows. Section~\ref{sec:background} reviews the CME and FSP framework. Section~\ref{sec:methodology} presents the flux-based pruning rule, adaptive time step selection, and CME matrix construction algorithm. Section~\ref{sec:theory} derives local and global error bounds. Section~\ref{sec:results} reports numerical experiments on four benchmark systems, including rigorous error analysis for a bottleneck reaction network, computational efficiency evaluation of the matrix reconstruction approach, and a head-to-head comparison with the Krylov-FSP-SSA method~\cite{Sidje2015}. Section~\ref{sec:discussion} discusses the results and Section~\ref{sec:conclusions} concludes with a summary and outlook.

\section{Background}
\label{sec:background}

We first briefly review the necessary definitions and methods, then state the main motivation of this paper in Remark~\ref{rem:stiffness-challenges}. 

\subsection{The Chemical Master Equation}

\begin{definition}[Chemical Master Equation (CME)]
The \textit{Chemical Master Equation (CME)} describes the time evolution of the probability distribution $p(\boldsymbol{x}, t)$ over the state $\boldsymbol{x}$. For a single state $\boldsymbol{x}$, the CME is
\begin{equation}
    \frac{d}{dt} p(\boldsymbol{x}, t) =
    \sum_{k=1}^{M} \big[\alpha_k(\boldsymbol{x} - \boldsymbol{\nu}_k)\, p(\boldsymbol{x} - \boldsymbol{\nu}_k, t)
    - \alpha_k(\boldsymbol{x})\, p(\boldsymbol{x}, t)\big],
\end{equation}
where $\alpha_k(\boldsymbol{x})$ is the propensity for reaction $k$ at state $\boldsymbol{x}$ and $\boldsymbol{\nu}_k$ is its stoichiometric change vector.

For the entire state space $\mathbf{X}$, the CME can be written in matrix form as
\begin{equation}
    \frac{d}{dt}\, \boldsymbol{p}(\mathbf{X}, t) = \mathbf{A}\, \boldsymbol{p}(\mathbf{X}, t),
\end{equation}
where $\boldsymbol{p}(\mathbf{X}, t)$ is the probability vector over $\mathbf{X}$ and the \textit{generator matrix} $\mathbf{A}$ has entries
\begin{align}
    \mathbf{A}_{ij} =
    \begin{cases}
        -\displaystyle\sum_{k=1}^{M} \alpha_k(\mathbf{X}_j), & \text{if } i=j, \\
         \alpha_k(\mathbf{X}_j) & \text{if }\, \mathbf{X}_i=\mathbf{X}_j+\boldsymbol{\nu}_k, \\
    \end{cases}
\end{align}
and $\mathbf{A}_{ij}=0$ when no reaction connects $\mathbf{X}_j$ to $\mathbf{X}_i$.
\end{definition}

\begin{remark}[Generator matrix properties]
The generator $\mathbf{A}$ satisfies:  
\begin{enumerate}
    \item $\mathbf{A}_{ij}\ge 0$ for $i\neq j$ (off-diagonal nonnegative),
    \item $\mathbf{A}_{ii}\le 0$ (diagonal nonpositive),
    \item $\sum_{i} \mathbf{A}_{ij}=0$ for all $j$ (each column sums to zero),
\end{enumerate}
\end{remark}

\subsection{Finite State Projection Method}

The Finite State Projection (FSP) method \cite{Munsky2006} approximates the CME by restricting dynamics to a finite subset $J \subset \mathbf{X}$. The complement $J' = \mathbf{X} \setminus J$ is aggregated into a single absorbing point (a sink state). Writing the full generator in block form with respect to $(J, J')$, the evolution of the projected distribution satisfies
\begin{equation}
    \dot{\boldsymbol{p}}_J(t) = \mathbf{A}_J\,\boldsymbol{p}_J(t) + \mathbf{A}_{JJ'}\,\boldsymbol{p}_{J'}(t),
    \label{eq:fsp-exact}
\end{equation}
where $\mathbf{A}_J$ is the principal submatrix of $\mathbf{A}$ on $J$ (transitions within $J$, exact), and $\mathbf{A}_{JJ'}$ couples $J'$ back into $J$. If all initial probability is placed inside $J$, then $\boldsymbol{p}_{J'}(0) = \mathbf{0}$ and the coupling term is initially zero. The FSP approximation drops this term entirely, yielding the closed finite-dimensional system
\begin{equation}
    \boldsymbol{p}(\mathbf{X}_J, t) \;\approx\; e^{\mathbf{A}_J t}\, \boldsymbol{p}_J(0).
    \label{eq:fsp-approx}
\end{equation}
Because $\mathbf{A}_J$ is a principal submatrix of a conservative generator, its column sums are nonpositive (transitions to $J'$ are missing from the off-diagonal, so probability leaks out at the boundary). Consequently $\Gamma_J(t_f) := \mathbf{1}^\top e^{\mathbf{A}_J t_f}\boldsymbol{p}_J(0) \le 1$ measures the probability mass retained in $J$ under the approximation. The following two theorems \cite{Munsky2006} give the theoretical guarantees.

\begin{theorem}[Monotonicity \cite{Munsky2006}]
\label{thm:fsp-monotone}
If $J \subseteq J'$, then $e^{\mathbf{A}_J t}\boldsymbol{p}_J(0) \le e^{\mathbf{A}_{J'} t}\boldsymbol{p}_{J'}(0)$ componentwise for all $t \ge 0$. That is, enlarging the projection set can only increase (or leave unchanged) the probability assigned to each state.
\end{theorem}

\begin{theorem}[FSP error bound \cite{Munsky2006}]
\label{thm:classical_fsp}
Let $\mathbf{A}_J$ be the principal submatrix of a CME generator on a finite set $J \subset \mathbf{X}$, and let $p_J(t)$ denote the true CME solution restricted to $J$. If
\begin{equation}
\Gamma_J(t_f) := \mathbf{1}^\top e^{\mathbf{A}_J t_f}\,\boldsymbol{p}_J(0) \;\ge\; 1-\varepsilon,
\label{eq:fsp-mass-cond}
\end{equation}
then the elementwise bounds
\begin{equation}
e^{\mathbf{A}_J t_f}\,\boldsymbol{p}_J(0) \;\le\; p_J(t_f) \;\le\; e^{\mathbf{A}_J t_f}\,\boldsymbol{p}_J(0) + \varepsilon\,\mathbf{1}
\label{eq:fsp-elementwise}
\end{equation}
hold for all $t_f \ge 0$.
\end{theorem}

The lower bound follows from monotonicity (Theorem~\ref{thm:fsp-monotone}): truncation can only reduce probability at each state. The upper bound says the pointwise error is at most $\varepsilon$, the total leakage through the boundary. Crucially, $\varepsilon = 1 - \Gamma_J(t_f)$ is computable without access to the true solution, so the FSP algorithm proceeds by iteratively enlarging $J$ until $\Gamma_J \ge 1-\varepsilon$ is verified.

\subsubsection{Time-stepping and sliding window approaches}

\begin{figure}[!htbp]
    \centering
    \includegraphics[width=\linewidth]{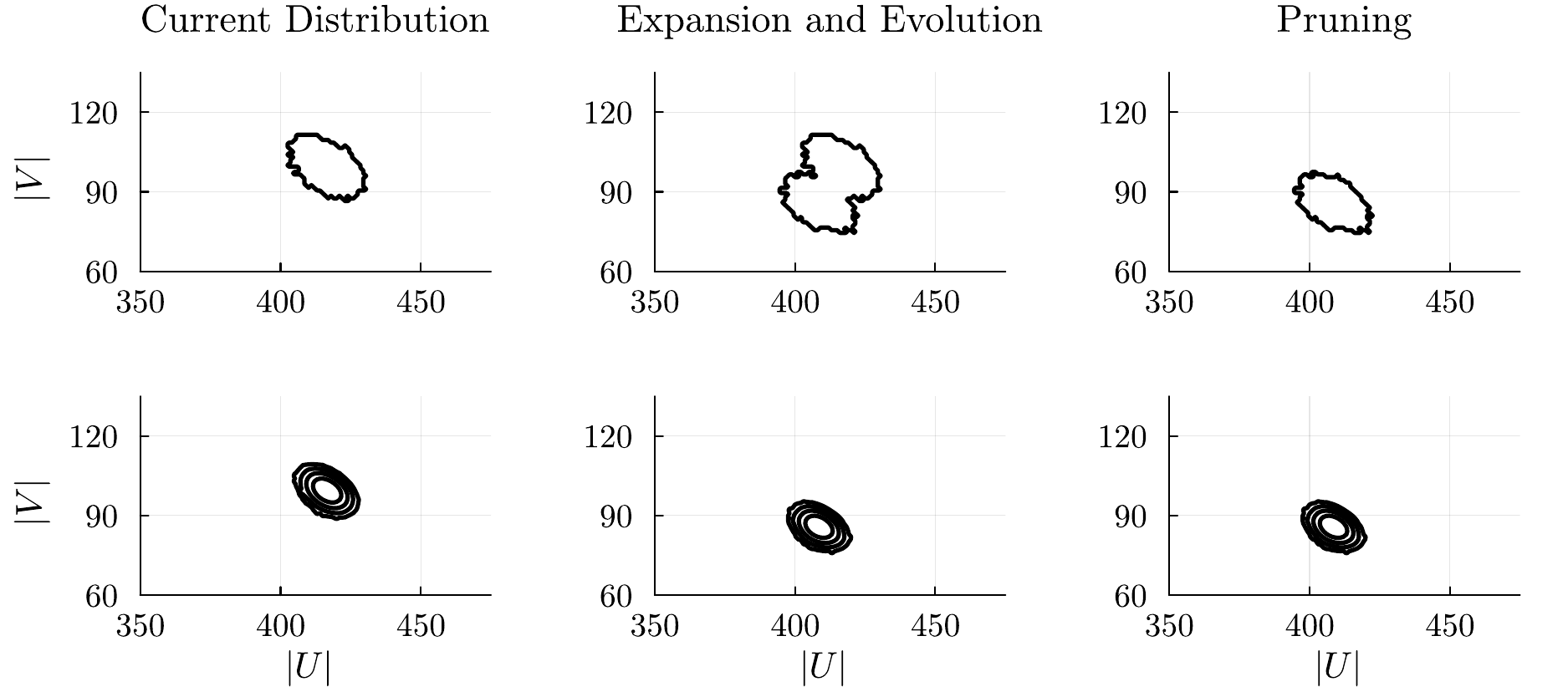}
    \caption{Visualization of the adaptive FSP procedure (species $U$ and $V$) during a single time step in Algorithm~\ref{alg:fsp-naive}. The top row shows the truncation boundary during the expand, evolve, and prune phases, while the bottom row shows the corresponding probability contours.}
    \label{fig:adapt-FSP}
\end{figure}

For simulations over long time intervals, maintaining a single fixed projection space is generally impractical: the reachable state space grows combinatorially in time, quickly making a static truncation too large to store or inefficient to evolve. A natural approach is to partition the time domain $[0,T]$ into subintervals $[t_k,t_{k+1}]$ and solve a sequence of restricted CME problems, each posed on a projection space $\mathbf{X}_J^{(k)}$ tailored to the states that are relevant during that subinterval \cite{Wolf2010, Dinh_2020}.  
At each time step, the method initializes $\boldsymbol{p}_J^{(k)}(t_k)$ from the previous solution, expands $\mathbf{X}_J^{(k)}$ to include newly reachable states, evolves the distribution forward in time, and prunes states whose probabilities fall below a chosen tolerance.

\begin{algorithm}[H]
\caption{FSP with Time Stepping and State Space Control}
\label{alg:fsp-naive}
\begin{algorithmic}[1]
\REQUIRE initial state $x_0$, time range $[t_0,t_f]$, propensities $\{\alpha_k(x)\}$, stoichiometries $\{\nu_k\}$, step $\Delta t$
\STATE $t \gets t_0$, \; $J \gets \{x_0\}$, \; $\boldsymbol{p} \gets [1]$
\WHILE{$t < t_f$}
  \STATE \textbf{Expand:} add neighbors of $J$ reachable by one (or preset $r$) reactions
  \STATE \textbf{Assemble:} build $\mathbf{A}_{JJ}$ on $J$
  \STATE \textbf{Evolve:} $\boldsymbol{p} \gets \exp(\mathbf{A}_{JJ}\Delta t)\,\boldsymbol{p}$
  \STATE \textbf{Prune:} remove states with $p(x)$ below a heuristic threshold
  \STATE $t \gets \min\{t+\Delta t,\,t_f\}$
\ENDWHILE
\STATE \textbf{return} $(J,\boldsymbol{p})$ at $t_f$
\end{algorithmic}
\end{algorithm}

This \emph{sliding-window} approach provides a general template for adaptive FSP methods with dynamically evolving state spaces \cite{Dinh2016}. Many variants differ in how expansion and pruning are performed, but they all follow the same broad framework, which is to enlarge the state space to capture new probability flow, evolve the distribution on the truncated generator, and prune regions that have become irrelevant. 
The remainder of this paper focuses on refining and analyzing the pruning and time step determination steps within this general adaptive FSP skeleton.

\begin{remark}[Stiffness and multiscale challenges in adaptive FSP]
\label{rem:stiffness-challenges}
Adaptive FSP methods must grapple with stiffness that arises when the generator $\mathbf{A}$ exhibits a wide separation of timescales across the state space, a characteristic of many biochemical reaction networks. This stiffness manifests in two coupled ways. First, \textit{spatial stiffness} emerges through connectivity bottlenecks. As illustrated in Figure~\ref{fig:bottleneck}, a bridge state $s_b$ with negligible probability but large exit rate $w(s_b) := \sum_{k=1}^M \alpha_k(s_b)$ (defined formally in Definition~\ref{def:flux}) creates a fast pathway between slow metastable regions, and pruning $s_b$ based on probability alone severs this connection, rendering the generator reducible and trapping probability mass in isolated components. Second, stiffness also manifests itself temporally. When probability rapidly redistributes during transient phases, the total flux $\Phi_{\mathrm{total}}(t) = \sum_{\boldsymbol{x}} p(\boldsymbol{x},t) w(\boldsymbol{x})$ becomes large and requires small steps $\Delta t$ for stability, while metastable slow periods with low flux permit large steps for efficiency (Figure~\ref{fig:temporal-stiffness}). These manifestations are coupled through the exit rates $w(\boldsymbol{x})$, which simultaneously create spatial bottlenecks and drive temporal flux variations. Our objective is to design pruning and time-stepping heuristics that jointly address both challenges, preserving network connectivity through flux-aware state selection while adapting $\Delta t$ to track the instantaneous activity level $\Phi_{\mathrm{total}}(t)$, thereby maintaining both structural integrity and computational tractability for robust multiscale simulation.

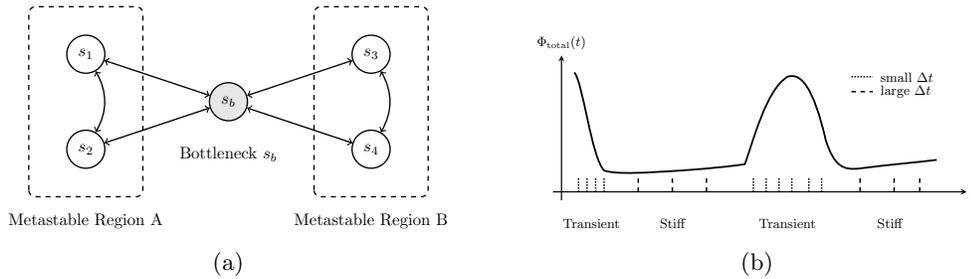
\begin{figure}[ht!]
\centering

\begin{subfigure}[b]{0.46\textwidth}
\centering
\resizebox{\linewidth}{!}{%
\begin{tikzpicture}[node distance=2.2cm, auto, thick,
  state/.style={circle, draw, minimum size=0.8cm}]
  \node[state] (A1) at (0,1) {$s_1$};
  \node[state] (A2) at (0,-1) {$s_2$};
  \node[state, fill=gray!20] (B) at (3,0) {$s_b$};
  \node[state] (C1) at (6,1) {$s_3$};
  \node[state] (C2) at (6,-1) {$s_4$};

  \path[<->, thin] (A1) edge[bend left] (A2);
  \path[<->, thin] (C1) edge[bend left] (C2);
  \path[<->, line width=1.8pt] (A1) edge
    node[above, sloped, font=\scriptsize] {$\Phi(s_b)=p(s_b)w(s_b)$} (B);
  \path[<->, line width=1.8pt] (A2) edge (B);
  \path[<->, line width=1.8pt] (C1) edge (B);
  \path[<->, line width=1.8pt] (C2) edge (B);

  \draw[dashed, rounded corners] (-1.2,-2) rectangle (1.2,2);
  \draw[dashed, rounded corners] (4.8,-2) rectangle (7.2,2);

  \node at (0, -2.5) {Metastable Region A};
  \node at (6, -2.5) {Metastable Region B};
  \node[font=\scriptsize] at (3, -0.85) {$p(s_b)\!\ll\!1,\;w(s_b)\!\gg\!0$};
\end{tikzpicture}%
}
\caption{}
\label{fig:bottleneck}
\end{subfigure}
\hfill
\begin{subfigure}[b]{0.46\textwidth}
\centering
\resizebox{\linewidth}{!}{%
\begin{tikzpicture}[thick]
  \draw[->, >=stealth, thick] (-0.2,0) -- (9.5,0) node[right, font=\small] {$t$};
  \draw[->, >=stealth, thick] (0,-0.2) -- (0,3.2) node[above, font=\small] {$\Phi_{\text{total}}(t)$};
  \draw[line width=1.2pt]
    (0.3,2.8) .. controls (0.5,2.6) and (0.7,0.8) .. (1.0,0.5)
    .. controls (1.3,0.4) and (2.0,0.45) .. (2.8,0.5)
    .. controls (3.5,0.55) and (4.0,0.6) .. (4.3,0.65)
    .. controls (4.5,1.2) and (4.8,2.4) .. (5.3,2.7)
    .. controls (5.6,2.8) and (5.9,2.5) .. (6.2,1.2)
    .. controls (6.4,0.6) and (6.6,0.5) .. (7.0,0.55)
    .. controls (7.8,0.65) and (8.4,0.7) .. (8.8,0.75);
  \foreach \x in {0.4,0.6,0.8,1.0,4.5,4.8,5.1,5.4,5.8,6.1}
    \draw[densely dotted, line width=1pt] (\x,0) -- (\x,0.35);
  \foreach \x in {1.8,2.6,3.4,7.0,7.8,8.4}
    \draw[dashed, line width=1pt] (\x,0) -- (\x,0.35);
  \node[font=\small] at (0.7,-0.75) {Transient};
  \node[font=\small] at (2.6,-0.75) {Stiff};
  \node[font=\small] at (5.3,-0.75) {Transient};
  \node[font=\small] at (7.7,-0.75) {Stiff};
  \begin{scope}[shift={(6.8,2.7)}]
    \draw[densely dotted, line width=1pt] (0,0) -- (0.5,0);
    \node[right, font=\small] at (0.55,0) {small $\Delta t$};
    \draw[dashed, line width=1pt] (0,-0.35) -- (0.5,-0.35);
    \node[right, font=\small] at (0.55,-0.35) {large $\Delta t$};
  \end{scope}
\end{tikzpicture}%
}
\caption{}
\label{fig:temporal-stiffness}
\end{subfigure}
\caption{Illustration of stiffness in both spatial and temporal domains. (\subref{fig:bottleneck}) shows a bottleneck state $s_b$ that connects metastable regions, while (\subref{fig:temporal-stiffness}) shows alternating transient and stiff regimes reflected in the total flux $\Phi_{\text{total}}(t)$.}
\label{fig:stiffness-combined}
\end{figure}
\end{remark}

\section{Methodology: Flux-Based Adaptive FSP}
\label{sec:methodology}

The adaptive FSP involves three coupled decisions at each step: which states to add (\emph{expansion}), how far to advance in time (\emph{time stepping}), and which states to remove (\emph{pruning}).  Each decision is governed by a distinct quantity.  Expansion adds neighboring states to $J$, driving the boundary outflux $\Phi_{\mathrm{out}}$ (the rate probability escapes the active set) toward zero.  Pruning removes low-probability states while protecting those with significant flux; the Stage~1 quantile guarantees the pruned mass is at most $\alpha$, bounding the boundary outflux to $\alpha\,w_{\max}(J')\,\Delta t$, while the Stage~2 flux criterion ($p(\boldsymbol{x})w(\boldsymbol{x}) \geq \varepsilon_{\mathrm{flux}}\,\Phi_{\mathrm{total}}$) ensures high-rate connector states are retained, keeping $w_{\max}(J')$ small.  Time stepping limits $\Phi_{\mathrm{total}}\,\Delta t \leq \varepsilon_{\Delta t}$, controlling the fraction of mass undergoing transitions per step.  Together these yield a global error accumulation of $2N\alpha + N(\bar{\Phi}_{\mathrm{out}}\,\Delta t + \varepsilon_{\mathrm{ODE}})$ over $N$ steps (Corollary~\ref{cor:global-error}, Section~\ref{sec:theory}), where the outflux term is driven to near zero by expansion, leaving $2\alpha$ as the dominant per-step cost.  This section develops each component in turn.

\begin{remark}[Total flux versus boundary outflux as the protection criterion]
\label{rem:total-flux-criterion}
One might ask why Algorithm~\ref{alg:flux-prune} uses total flux $\Phi(\boldsymbol{x},t) = p(\boldsymbol{x},t)\,w(\boldsymbol{x})$ rather than boundary outflux $\Phi_\partial(\boldsymbol{x},t) := p(\boldsymbol{x},t)\sum_{\boldsymbol{y}\in J'}\mathbf{A}_{\boldsymbol{y}\boldsymbol{x}}$ to decide which states to protect.
When a bottleneck state $s_b$ first enters $J$, its destination states are not yet expanded into $J$, so $\Phi_\partial(s_b,t)$ is large and boundary-outflux protection would retain it.
However, once the destination regions are brought into $J$ by expansion (they have high probability and survive the quantile filter), the boundary outflux of $s_b$ drops to approximately zero, since its transitions now lead entirely within $J$.
At that mature stage, boundary-outflux protection would fail to retain $s_b$ despite its structural importance as a connector.
Total flux $p(s_b)\,w(s_b)$ is large throughout both stages because $w(s_b)$ is intrinsically large, reflecting the high rate at which probability cycles \emph{through} $s_b$ regardless of whether the destinations are in $J$ or $J'$.
An additional practical advantage is that total flux $p(\boldsymbol{x})\cdot(-\mathbf{A}_{\boldsymbol{x}\boldsymbol{x}})$ is a byproduct of reading the generator diagonal, requiring no knowledge of which states will be pruned.
\end{remark}

\begin{definition}[Outgoing probability flux]
\label{def:flux}
Let $w(\boldsymbol{x}) := \sum_{k=1}^M \alpha_k(\boldsymbol{x}) = -\mathbf{A}_{\boldsymbol{x},\boldsymbol{x}}$ denote the total exit rate from state $\boldsymbol{x}$. The \textit{outgoing probability flux} from state $\boldsymbol{x}$ at time $t$ is
\begin{equation}
\Phi(\boldsymbol{x},t) := p(\boldsymbol{x},t) \cdot w(\boldsymbol{x}).
\end{equation}
\end{definition}

The flux $\Phi(\boldsymbol{x},t)$ measures the instantaneous rate at which probability mass exits state $\boldsymbol{x}$. For small $\Delta t > 0$, approximately $\Phi(\boldsymbol{x},t) \cdot \Delta t$ units of probability mass leave $\boldsymbol{x}$ during the interval $[t, t+\Delta t]$. Crucially, states with large exit rates $w(\boldsymbol{x})$ but low probabilities $p(\boldsymbol{x},t)$ can maintain significant flux, acting as critical connections for probability flow between regions of state space despite being rarely occupied.

\begin{definition}[Total system flux]
The total probability flux of the system at time $t$ over state set $J$ is
\begin{equation}
\Phi_{\mathrm{total}}(J,t) := \sum_{\boldsymbol{x} \in J} \Phi(\boldsymbol{x},t) = \sum_{\boldsymbol{x} \in J} p(\boldsymbol{x},t) \cdot w(\boldsymbol{x}).
\end{equation}
\end{definition}

\begin{proposition}[Total flux properties]
\label{prop:total-flux}
For the full state space $\mathcal{X}$ and any time $t \geq 0$, the total system flux satisfies
\begin{equation}
\Phi_{\mathrm{total}}(\mathcal{X},t) = \sum_{\boldsymbol{x} \in \mathcal{X}} p(\boldsymbol{x},t) \cdot w(\boldsymbol{x}) = -\boldsymbol{p}(t)^\top \mathrm{diag}(\mathbf{A}).
\end{equation}
Since $\sum_{\boldsymbol{x}} p(\boldsymbol{x},t) = 1$ for all $t$, the total flux represents a weighted average of exit rates across the state space. For reaction networks with mass-action or similar kinetics where exit rates grow at most polynomially with state values, $\Phi_{\mathrm{total}}(\mathcal{X},t)$ remains bounded on compact time intervals and tracks the instantaneous activity level of the system.
\end{proposition}

\begin{proof}
By definition of the generator matrix,
\begin{align}
\Phi_{\mathrm{total}}(\mathcal{X},t) &= \sum_{\boldsymbol{x} \in \mathcal{X}} p(\boldsymbol{x},t) \cdot (-\mathbf{A}_{\boldsymbol{x},\boldsymbol{x}}) = -\boldsymbol{p}(t)^\top \mathrm{diag}(\mathbf{A}).
\end{align}
Since the diagonal entries $\mathbf{A}_{\boldsymbol{x},\boldsymbol{x}} = -\sum_{k=1}^M \alpha_k(\boldsymbol{x})$ depend only on the reaction network structure and state $\boldsymbol{x}$, the total flux evolves with the probability distribution. 
\end{proof}

\subsection{Flux and Network Connectivity}
\label{sec:flux-connectivity}

Having defined flux, we briefly review how it connects to the structural properties governing long-time behavior. This connection motivates our pruning rule, though the error bounds in Section~\ref{sec:theory} will take a more direct approach based on system fluxes. A generator $\mathbf{A}$ on a finite state space $J$ is \textit{irreducible} if the corresponding state-space graph is strongly connected. Equivalently (with our column convention, where $A_{x,y}>0$ is a transition $y\!\to\!x$), there is no nontrivial partition $J=J_1\cup J_2$ with $J_1\cap J_2=\emptyset$ such that $A_{J_1,J_2}=0$ (no transitions from $J_2$ into $J_1$). Irreducibility ensures that the chain can reach any state from any other in finite expected time and implies a unique stationary distribution $\boldsymbol{\pi}\gg 0$ for the finite continuous time Markov chains. The spectral gap
\[
\lambda := \min\{\operatorname{Re}(-\lambda_k(\mathbf{A})):\lambda_k(\mathbf{A})\neq 0\} > 0
\]
governs exponential convergence to stationarity (e.g., in $\pi$-weighted norms), i.e., $\|\boldsymbol{p}(t)-\boldsymbol{\pi}\|\le C e^{-\lambda t}$. The connection between flux and connectivity is formalized through Cheeger's inequality for continuous-time Markov chains \cite{Diaconis1991,63529}. Define the \textit{conductance} of a subset $S \subset J$ as
\begin{equation}
\Psi(S) := \frac{\sum_{\boldsymbol{x} \in S}\sum_{\boldsymbol{y} \in J \setminus S} \pi(\boldsymbol{x}) \mathbf{A}_{\boldsymbol{y},\boldsymbol{x}}}{\min\{\pi(S), \pi(J\setminus S)\}},
\label{eq:cheeger}
\end{equation}
where $\pi(S) := \sum_{\boldsymbol{x} \in S} \pi(\boldsymbol{x})$. Let $\Psi_* := \min_{S: 0 < \pi(S) \leq 1/2} \Psi(S)$ denote the minimum conductance, which quantifies the flux across the most restrictive bottleneck. Cheeger's inequality establishes the two-sided bound
\begin{equation}
\label{eq:cheeger1}
\frac{\Psi_*^2}{2} \leq \lambda \leq 2\Psi_*,
\end{equation}
linking the spectral gap to the minimum conductance.

\begin{figure}[H]
\centering
\includegraphics[width=0.72\linewidth]{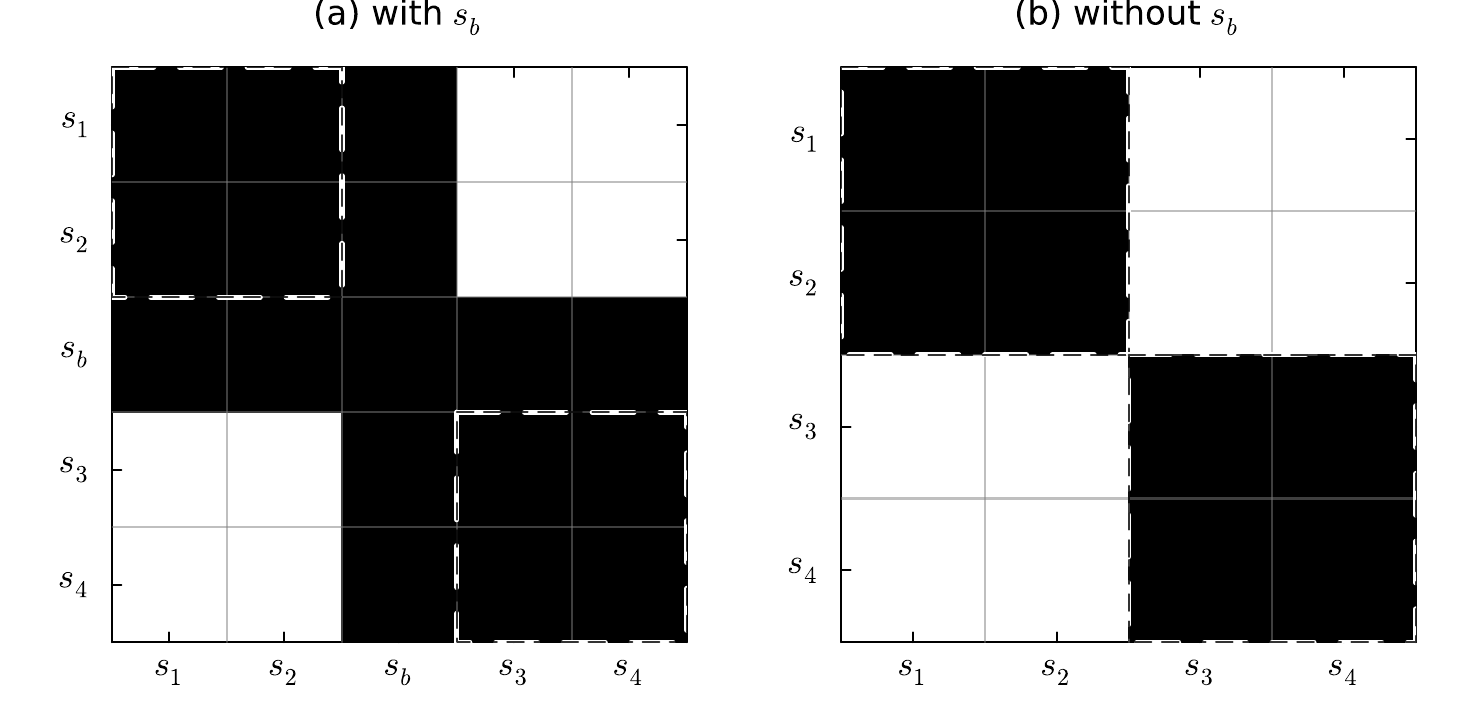}
\caption{Sparsity pattern of the generator $\mathbf{A}$ for the bottleneck network (Figure~\ref{fig:bottleneck}), ordered as $[s_1,s_2,s_b,s_3,s_4]$. Black cells are nonzero entries; white cells are zero. Dashed boxes mark the two metastable regions. \textbf{(a)} With $s_b$: the $s_b$ row and column (index 3) bridge both regions, making $\mathbf{A}$ irreducible. \textbf{(b)} Without $s_b$: the generator is exactly block-diagonal; the two metastable regions are disconnected and probability cannot flow between them.}
\label{fig:generator-sparsity}
\end{figure}

Our pruning rule protects states with large instantaneous exit activity $\Phi(\boldsymbol{x},t)=p(\boldsymbol{x},t)\,w(\boldsymbol{x})$. The connection to cut conductance is through the bound: for any $S\subset J$ and $\boldsymbol{x}\in S$,
\begin{equation}
\label{eq:flux-cut-bound}
\sum_{\boldsymbol{y}\notin S}\pi(\boldsymbol{x})\,\mathbf{A}_{\boldsymbol{y},\boldsymbol{x}} \;\le\; \pi(\boldsymbol{x})\,w(\boldsymbol{x}),
\end{equation}
so the stationary exit activity $\pi(\boldsymbol{x})w(\boldsymbol{x})$ upper-bounds the contribution of $\boldsymbol{x}$ to any cut flux. Thus, states with small probability but large $w(\boldsymbol{x})$ may contribute disproportionately to cut numerators. Pruning such states can reduce the flux across critical cuts and thereby degrade conductance, particularly when they act as bridges between metastable regions. In the extreme case where all cross-region paths are severed, $\Psi_*=0$ and irreducibility is lost entirely. This is the mechanism illustrated in Figure~\ref{fig:generator-sparsity}: $s_b$ has small occupancy yet large $w(s_b)$, so removing it collapses the cut flux to zero and disconnects the generator. These observations motivate the flux metric as a proxy for conductance-relevant activity. States with small $p(\boldsymbol{x},t)$ but large $w(\boldsymbol{x})$ have the capacity to sustain large cut flux, and protecting them is expected to reduce spectral degradation under truncation.

\subsubsection{Flux-Aware Truncation Algorithm}

We propose a two-stage pruning approach that combines probability-based candidate selection with flux-based protection. The algorithm identifies states for potential removal using a quantile threshold, then protects low-probability states that carry significant probability flow.

\begin{algorithm}[H]
\caption{Flux-Preserving Pruning}
\label{alg:flux-prune}
\begin{algorithmic}[1]
\REQUIRE Active set $\mathcal{S}$, probability vector $\boldsymbol{p}$, generator $\mathbf{A}$, quantile tolerance $\alpha \in (0,1)$, flux tolerance $\varepsilon_{\mathrm{flux}} > 0$
\ENSURE Updated set $\mathcal{S}_{\mathrm{new}}$ and renormalized probability $\boldsymbol{p}_{\mathrm{new}}$
\STATE
\STATE \textit{Stage 1: Candidate Selection}
\STATE Sort states $\boldsymbol{x} \in \mathcal{S}$ in ascending order of $p(\boldsymbol{x}, t)$
\STATE Define candidate set $\mathcal{C} \subseteq \mathcal{S}$ as the minimal set satisfying
\begin{equation}
\sum_{\boldsymbol{x} \in \mathcal{C}} p(\boldsymbol{x}, t) \leq \alpha \quad \text{and} \quad p(\boldsymbol{x}, t) \leq p(\boldsymbol{y}, t) \; \forall \boldsymbol{x} \in \mathcal{C}, \boldsymbol{y} \in \mathcal{S} \setminus \mathcal{C}
\end{equation}
\STATE \textit{Stage 2: Flux-Based Protection}
\FOR{each $\boldsymbol{x} \in \mathcal{S}$}
    \STATE Compute exit rate: $w(\boldsymbol{x}) = -\mathbf{A}_{\boldsymbol{x},\boldsymbol{x}}$
    \STATE Compute flux: $\Phi(\boldsymbol{x}, t) = p(\boldsymbol{x}, t) \cdot w(\boldsymbol{x})$
\ENDFOR
\STATE Compute total flux: $\Phi_{\mathrm{total}} = \sum_{\boldsymbol{x} \in \mathcal{S}} \Phi(\boldsymbol{x}, t)$
\STATE Set flux threshold: $\Phi_{\mathrm{thr}} = \varepsilon_{\mathrm{flux}} \cdot \Phi_{\mathrm{total}}$
\STATE Define protected set: $\mathcal{P} = \{\boldsymbol{x} \in \mathcal{C} : \Phi(\boldsymbol{x}, t) \geq \Phi_{\mathrm{thr}}\}$
\STATE Define pruning set: $\mathcal{S}_{\mathrm{prune}} = \mathcal{C} \setminus \mathcal{P}$
\STATE
\STATE \textit{Stage 3: Update and Renormalization}
\STATE $\mathcal{S}_{\mathrm{new}} \gets \mathcal{S} \setminus \mathcal{S}_{\mathrm{prune}}$
\STATE $\boldsymbol{p}_{\mathrm{new}} \gets \boldsymbol{p}|_{\mathcal{S}_{\mathrm{new}}} / \|\boldsymbol{p}|_{\mathcal{S}_{\mathrm{new}}}\|_1$ \quad (restrict and renormalize)
\RETURN $\mathcal{S}_{\mathrm{new}}, \boldsymbol{p}_{\mathrm{new}}$
\end{algorithmic}
\end{algorithm}

Stage 1 identifies low-probability candidate states for removal based on the quantile threshold $\alpha$. Without this candidate selection, the method would need to evaluate flux for every state at every pruning step, incurring unnecessary computational cost when most high-probability states are never at risk of removal. The quantile filter focuses computational effort on the tail of the distribution where pruning decisions are nontrivial. Stage 2 protects states in $\mathcal{C}$ that carry significant flux relative to total system flux. Consider a bottleneck state $s_b$ connecting metastable regions as seen in Figure~\ref{fig:bottleneck}. Such states transmit substantial probability flow between regions despite low occupancy. Removing them degrades conductance, shrinks the spectral gap, and harms mixing. The flux tolerance $\varepsilon_{\mathrm{flux}}$ controls the trade-off between compression and connectivity preservation. Stage 3 removes unprotected candidates and renormalizes the probability distribution to maintain unit total mass.

\subsection{Adaptive Time Stepping}

Stochastic reaction networks commonly exhibit \emph{temporal stiffness}. The system passes through phases with widely separated characteristic timescales, and the appropriate time step differs by orders of magnitude across these phases.  During a slow transient, when probability is concentrated in low-propensity states, large steps can be taken safely without missing important dynamics.  During a stiff or burst phase, when reactions fire rapidly and probability redistributes quickly, small steps are required to maintain accuracy.  A fixed step size must be calibrated to the worst case and is therefore grossly conservative during slow phases, leading to unnecessary computation.

A natural measure of the current system activity is $\Phi_{\mathrm{total}}(J_n,t_n) = \sum_{\boldsymbol{x}\in J_n}p(\boldsymbol{x},t_n)\,w(\boldsymbol{x})$, the probability-weighted mean propensity.  $\Phi_{\mathrm{total}}$ tracks the actual rate of probability movement under the current distribution. It is small when probability sits in slow states and large when probability occupies fast states.  The adaptive step $\Delta t_n = \varepsilon_{\Delta t}/\Phi_{\mathrm{total}}$ therefore contracts during stiff phases and expands during slow transients automatically, matching the step size to the active timescale at each iteration.

\begin{proposition}[Adaptive time step rule]
\label{prop:flux-timestep}
The step size
\begin{equation}
\label{eq:flux-timestep}
\Delta t_n = \frac{\varepsilon_{\Delta t}}{\Phi_{\mathrm{total}}(J_n,t_n)}, \qquad \Phi_{\mathrm{total}}(J_n,t_n) := \sum_{\boldsymbol{x}\in J_n} p(\boldsymbol{x},t_n)\,w(\boldsymbol{x}),
\end{equation}
limits the leading-order fraction of probability mass redistributed within the active set during one step to at most $\varepsilon_{\Delta t} + O(\Delta t_n^2)$, and automatically shortens steps during fast or stiff phases while lengthening them during slow transients.
\end{proposition}
\begin{proof}
Let $\widehat{\boldsymbol{p}}(t)$ denote the compressed dynamics on $J_n$, so that $\widehat{\boldsymbol{p}}(t_n+\Delta t_n) = e^{\widetilde{\mathbf{A}}_{JJ}\Delta t_n}\widetilde{\boldsymbol{p}}_n$. Expanding to first order gives
\[
\widehat{\boldsymbol{p}}(t_n+\Delta t_n) - \widetilde{\boldsymbol{p}}_n = \widetilde{\mathbf{A}}_{JJ}\widetilde{\boldsymbol{p}}_n\,\Delta t_n + O(\Delta t_n^2).
\]
Define the retained-state exit rate
\[
\lambda_J(\boldsymbol{x}) := -\widetilde{A}_{\boldsymbol{x}\boldsymbol{x}} = \sum_{\substack{\boldsymbol{y}\in J_n\\ \boldsymbol{y}\neq \boldsymbol{x}}} \widetilde{A}_{\boldsymbol{y}\boldsymbol{x}},
\]
which satisfies $\lambda_J(\boldsymbol{x}) \leq w(\boldsymbol{x})$ because transitions from $\boldsymbol{x}$ to states outside $J_n$ are removed in the compressed generator. Using the non-negative off-diagonal entries of $\widetilde{\mathbf{A}}_{JJ}$ and the triangle inequality,
\begin{align*}
\|\widetilde{\mathbf{A}}_{JJ}\widetilde{\boldsymbol{p}}_n\|_1
&\leq \sum_{\boldsymbol{x}\in J_n}\lambda_J(\boldsymbol{x})\,\widetilde{p}_n(\boldsymbol{x})
+ \sum_{\boldsymbol{x}\in J_n}\sum_{\substack{\boldsymbol{y}\in J_n\\ \boldsymbol{y}\neq \boldsymbol{x}}}\widetilde{A}_{\boldsymbol{x}\boldsymbol{y}}\,\widetilde{p}_n(\boldsymbol{y}) \\
&= 2\sum_{\boldsymbol{x}\in J_n}\lambda_J(\boldsymbol{x})\,\widetilde{p}_n(\boldsymbol{x}) \\
&\leq 2\sum_{\boldsymbol{x}\in J_n}w(\boldsymbol{x})\,\widetilde{p}_n(\boldsymbol{x}) \\
&= 2\,\Phi_{\mathrm{total}}(J_n,t_n).
\end{align*}
Therefore
\[
\|\widehat{\boldsymbol{p}}(t_n+\Delta t_n) - \widetilde{\boldsymbol{p}}_n\|_1
\leq 2\,\Phi_{\mathrm{total}}(J_n,t_n)\,\Delta t_n + O(\Delta t_n^2),
\]
so the total variation distance, and hence the fraction of probability mass redistributed within the active set during the step, is bounded by
\[
\tfrac{1}{2}\|\widehat{\boldsymbol{p}}(t_n+\Delta t_n) - \widetilde{\boldsymbol{p}}_n\|_1
\leq \Phi_{\mathrm{total}}(J_n,t_n)\,\Delta t_n + O(\Delta t_n^2).
\]
The choice $\Delta t_n = \varepsilon_{\Delta t}/\Phi_{\mathrm{total}}$ therefore bounds this fraction by $\varepsilon_{\Delta t} + O(\Delta t_n^2)$. The adaptive property follows because $\Phi_{\mathrm{total}} = \mathbb{E}_p[w(\boldsymbol{x})]$ tracks the active timescale.
\begin{itemize}
\item \emph{Slow or stiff phase.} Probability is concentrated in low-propensity states, so $\Phi_{\mathrm{total}}$ is small and $\Delta t_n$ is large.
\item \emph{Fast or transient phase.} Probability shifts into high-propensity states, $\Phi_{\mathrm{total}}$ grows, and $\Delta t_n$ shrinks automatically without explicit phase detection.
\end{itemize}
\end{proof}

In practice, $\Phi_{\mathrm{total}}$ is computed as a byproduct of the flux computation in Algorithm~\ref{alg:flux-prune}, so the step size $\Delta t_n$ is available at no additional cost before each matrix exponential evaluation.

\subsection{Master Equation Matrix Construction}
\label{sec:matrix_construction}

The generator matrix $\mathbf{A}(t)$ must be constructed at each adaptive time step as the state space evolves. We employ a forward enumeration approach that exploits the sparse reaction connectivity structure to build the matrix efficiently. For a given state set $\mathcal{S}$ with $n = |\mathcal{S}|$ states and $R$ reaction channels, we construct $\mathbf{A}$ column by column. For each state $x \in \mathcal{S}$ and reaction $k$ with stoichiometry $\boldsymbol{\nu}_k$, we compute the destination state $y = x + \boldsymbol{\nu}_k$. If $y \in \mathcal{S}$, we evaluate the propensity $\alpha_k(x)$ and place it at position $(i_y, i_x)$ in the matrix, where $i_y$ and $i_x$ are the indices of states $y$ and $x$ in $\mathcal{S}$. We accumulate these off-diagonal entries to compute the column sum, then set the diagonal entry to the negative of this sum to ensure the generator property (zero column sums).

\begin{algorithm}[H]
\caption{Forward Enumeration Matrix Construction}
\label{alg:matrix_construction}
\begin{algorithmic}[1]
\REQUIRE State set $\mathcal{S}$, reaction model with stoichiometries $\{\boldsymbol{\nu}_k\}$ and propensities $\{\alpha_k\}$
\ENSURE Generator matrix $\mathbf{A}$
\STATE Initialize sparse matrix arrays: $I \leftarrow []$, $J \leftarrow []$, $V \leftarrow []$
\STATE Create index mapping: $\texttt{state\_id}: \mathcal{S} \rightarrow \{1, \ldots, |\mathcal{S}|\}$
\FOR{each state $x \in \mathcal{S}$ with index $j$}
    \STATE $\text{col\_sum} \leftarrow 0$
    \FOR{each reaction $k$ with stoichiometry $\boldsymbol{\nu}_k$}
        \STATE $y \leftarrow x + \boldsymbol{\nu}_k$
        \IF{$y \in \mathcal{S}$}
            \STATE $\alpha \leftarrow \texttt{propensity}_k(x)$
            \IF{$\alpha > 0$}
                \STATE $i \leftarrow \texttt{state\_id}[y]$
                \STATE Append $(i, j, \alpha)$ to sparse arrays
                \STATE $\text{col\_sum} \leftarrow \text{col\_sum} + \alpha$
            \ENDIF
        \ENDIF
    \ENDFOR
    \IF{$\text{col\_sum} > 0$}
        \STATE Append $(j, j, -\text{col\_sum})$ to sparse arrays \COMMENT{Diagonal entry}
    \ENDIF
\ENDFOR
\STATE $\mathbf{A} \leftarrow \texttt{sparse}(I, J, V, |\mathcal{S}|, |\mathcal{S}|)$
\RETURN $\mathbf{A}$
\end{algorithmic}
\end{algorithm}

The computational complexity is $O(|\mathcal{S}| \cdot R)$ for iterating through states and reactions. This exploits the sparse structure of the reaction graph as each state connects to at most $R$ other states through direct chemical transformations, avoiding the $O(|\mathcal{S}|^2)$ cost of checking all possible state pairs. For typical chemical reaction networks with $R \ll |\mathcal{S}|$, this forward enumeration approach yields substantial performance improvements over standard backward expansion methods.

\subsection{Full Algorithm}
\label{sec:full-algorithm}
This section presents the complete flux-adaptive FSP method integrating the components developed above including boundary expansion, forward enumeration matrix construction (Algorithm~\ref{alg:matrix_construction}), flux-based adaptive time stepping (Eq.~\eqref{eq:flux-timestep}), and connectivity-preserving pruning (Algorithm~\ref{alg:flux-prune}). The algorithm maintains a dynamically adapted state space $\mathcal{S}$ that expands to capture emerging probability mass and contracts via flux-aware pruning to control computational cost while preserving network connectivity.

\begin{algorithm}[H]
\caption{Flux-Adaptive FSP}
\label{alg:fsp-flux-adaptive}
\begin{algorithmic}[1]
\REQUIRE initial state $x_0$, time range $[t_0,t_f]$, propensities $\{\alpha_k(x)\}$, stoichiometries $\{\nu_k\}$, quantile tolerance $\alpha$, flux tolerance $\varepsilon_{\mathrm{flux}}$, time step tolerance $\varepsilon_{\Delta t}$
\STATE $t \gets t_0$, \; $\mathcal{S} \gets \{x_0\}$, \; $\boldsymbol{p} \gets [1]$, \; $\mathbf{A} \gets$ initial generator
\WHILE{$t < t_f$}
  \STATE \textbf{Expand:} $\mathcal{S}_{\text{new}} \gets \mathcal{S} \cup \partial \mathcal{S}$ (expand state space)
  \STATE \textbf{Construct:} $\mathbf{A}_{\text{new}} \gets$ forward enumeration (Algorithm~\ref{alg:matrix_construction})
  \STATE \textbf{Adapt step:} Compute $\Phi_{\mathrm{total}} = \sum_{x \in \mathcal{S}_{\text{new}}} p(x,t) \cdot w(x)$ and set $\Delta t = \varepsilon_{\Delta t}/\Phi_{\mathrm{total}}$ (Eq.~\eqref{eq:flux-timestep})
  \STATE \textbf{Evolve:} $\boldsymbol{p} \gets \exp(\mathbf{A}_{\text{new}}\Delta t)\,\boldsymbol{p}$
  \STATE \textbf{Prune:} $(\mathcal{S}_{\text{pruned}}, \boldsymbol{p}_{\text{pruned}}) \gets$ flux-preserving pruning (Algorithm~\ref{alg:flux-prune}) with parameters $\alpha, \varepsilon_{\mathrm{flux}}$
  \STATE $\mathcal{S} \gets \mathcal{S}_{\text{pruned}}$, \; $\boldsymbol{p} \gets \boldsymbol{p}_{\text{pruned}}$, \; $\mathbf{A} \gets \mathbf{A}_{\text{new}}|_{\mathcal{S}_{\text{pruned}}}$
  \STATE $t \gets \min\{t+\Delta t,\,t_f\}$
\ENDWHILE
\STATE \textbf{return} $(\mathcal{S},\boldsymbol{p})$ at $t_f$
\end{algorithmic}
\end{algorithm}

The computational bottleneck is the matrix exponential action $\exp(\mathbf{A}_{\text{new}}\Delta t)\boldsymbol{p}$ in the evolution step. We employ Krylov subspace methods via the EXPOKIT package \cite{Sidje1998}, which computes matrix exponential actions without forming the full exponential matrix. For a sparse generator $\mathbf{A}$ of dimension $n = |\mathcal{S}|$ with
$\mathrm{nnz}(\mathbf{A}) = O(n)$, EXPOKIT's Krylov--Arnoldi routine for
computing $e^{t\mathbf{A}} v$ has per-step cost
\[
  O\bigl(m\,\mathrm{nnz}(\mathbf{A}) + m^2 n\bigr) = O(m^2 n),
\]
where $m \ll n$ is the Krylov subspace dimension. In our setting, the
expansion and pruning operations scale as $O(nR)$, and constructing
$\mathbf{A}$ by forward enumeration also costs $O(nR)$, where $R$ is the
number of reactions. Since $m$ and $R$ are small constants independent of
$n$, the total per-step cost is effectively linear in the state-space size,
making simulations with $10^4$--$10^5$ states computationally tractable.

\section{Error Analysis}
\label{sec:theory}

We derive local and global error bounds for the adaptive FSP method based on probability fluxes across the truncation boundary. Let $\mathcal{X}$ denote the full state space and $J \subset \mathcal{X}$ be a finite active set. Define the complement $J' := \mathcal{X} \setminus J$ with $J \cap J' = \emptyset$ and $J \cup J' = \mathcal{X}$. The generator matrix admits a natural block decomposition:

\begin{equation}
\label{eq:block-partition}
\mathbf{A}=
\begin{pmatrix}
\mathbf{A}_{JJ} & \mathbf{A}_{JJ'}\\
\mathbf{A}_{J'J} & \mathbf{A}_{J'J'}
\end{pmatrix}, 
\qquad
\boldsymbol{p}(t)=
\begin{pmatrix}
\boldsymbol{p}_J(t)\\
\boldsymbol{p}_{J'}(t)
\end{pmatrix},
\end{equation}
Recall that under the column convention $\mathbf{A}_{ij} > 0$ represents a transition rate \emph{from} state $j$ \emph{to} state $i$. Accordingly, the block subscripts follow matrix index order (rows, columns): $\mathbf{A}_{JJ} \in \mathbb{R}^{|J| \times |J|}$ has rows and columns in $J$ and governs transitions within $J$; $\mathbf{A}_{J'J} \in \mathbb{R}^{|J'| \times |J|}$ has rows in $J'$ and columns in $J$, so its entries $(\mathbf{A}_{J'J})_{\boldsymbol{y}\boldsymbol{x}} = \mathbf{A}_{\boldsymbol{y}\boldsymbol{x}}$ for $\boldsymbol{y} \in J'$, $\boldsymbol{x} \in J$ describe transitions \emph{from} $J$ \emph{to} $J'$ (outflux); symmetrically, $\mathbf{A}_{JJ'} \in \mathbb{R}^{|J| \times |J'|}$ describes transitions \emph{from} $J'$ \emph{to} $J$ (influx). See Figure~\ref{fig:transitions}. The partitioned CME becomes:
\begin{equation}
\label{eq:partitioned-cme}
\frac{d}{dt}
\begin{pmatrix}
\boldsymbol{p}_J(t) \\
\boldsymbol{p}_{J'}(t)
\end{pmatrix}
=
\begin{pmatrix}
\mathbf{A}_{JJ} & \mathbf{A}_{JJ'}\\
\mathbf{A}_{J'J} & \mathbf{A}_{J'J'}
\end{pmatrix}
\begin{pmatrix}
\boldsymbol{p}_J(t) \\
\boldsymbol{p}_{J'}(t)
\end{pmatrix},
\end{equation}
where
\begin{equation}
\label{eq:row-sum-identities}
\mathbf{1}^\top \mathbf{A}_{JJ} = -\mathbf{1}^\top \mathbf{A}_{J'J}, \qquad \mathbf{1}^\top \mathbf{A}_{J'J'} = -\mathbf{1}^\top \mathbf{A}_{JJ'}.
\end{equation}
The column sum property of the generator $\mathbf{A}$ immediately yields the identities in Eq.~\eqref{eq:row-sum-identities}.
\begin{figure}[ht!]
    \centering
    \begin{tikzpicture}[node distance=3cm, auto, thick]
    \tikzstyle{state} = [circle, draw, minimum size=1cm, text centered]
    \tikzstyle{transition} = [->, thick]
    \node[state] (S) at (0, 0) {$J$};
    \node[state] (R) at (4, 0) {$J'$};
    \path[transition, loop left] (S) edge node[left] {$\mathbf{A}_{JJ}$} (S);
    \path[transition, loop right] (R) edge node[right] {$\mathbf{A}_{J'J'}$} (R);
    \path[transition] (S) edge[bend left=20] node[above] {$\mathbf{A}_{J'J}$} (R);
    \path[transition] (R) edge[bend left=20] node[below] {$\mathbf{A}_{JJ'}$} (S);
    \end{tikzpicture}
    \caption{Partitioned master equation transitions between the active set $J$ and its complement $J'$.}
    \label{fig:transitions}
\end{figure}
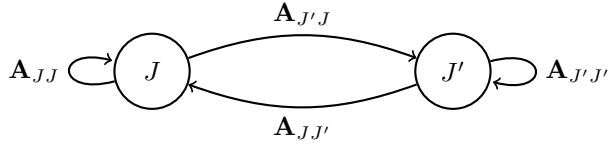

We define the boundary flux rates as
\begin{equation}
\label{eq:flux-defs}
\Phi_{\mathrm{out}}(t) := \mathbf{1}^\top \mathbf{A}_{J'J}\boldsymbol{p}_J(t), \qquad \Phi_{\mathrm{in}}(t) := \mathbf{1}^\top \mathbf{A}_{JJ'}\boldsymbol{p}_{J'}(t),
\end{equation}
where $\Phi_{\mathrm{out}}(t)$ is the total probability mass per unit time flowing from $J$ to $J'$ (the block $\mathbf{A}_{J'J}$ acts on $\boldsymbol{p}_J$ since its columns are indexed by $J$ and its entries give rates \emph{into} $J'$), and $\Phi_{\mathrm{in}}(t)$ is the mass flowing from $J'$ to $J$ (the block $\mathbf{A}_{JJ'}$ acts on $\boldsymbol{p}_{J'}$ since its columns are indexed by $J'$ and its entries give rates \emph{into} $J$). Both quantities are non-negative. The original FSP approach defines an absorbing sink state that collapses the complement of the truncated set. Following the block formulation, it is equivalent to the \textit{compressed generator} formed by adding the outgoing probability flux back to the diagonal elements:
\begin{equation}
\label{eq:compressed-generator}
\widetilde{\mathbf{A}}_{JJ} := \mathbf{A}_{JJ} + \diag(\mathbf{1}^\top \mathbf{A}_{J'J}),
\end{equation}
where $\operatorname{diag}(.)$ embeds a vector into the diagonal of a matrix. This operation ensures that any probability mass that would have left $J$ (as governed by $\mathbf{A}_{J'J}$) is instantaneously reflected back to its state of origin, creating a reflecting boundary condition.

\begin{proposition}[Compressed generator properties]
\label{prop:compressed-generator}
The compressed generator $\widetilde{\mathbf{A}}_{JJ}$ satisfies:
\begin{enumerate}
\item $\mathbf{1}^\top \widetilde{\mathbf{A}}_{JJ} = \mathbf{0}^\top$ (zero column sums, hence mass conservation),
\item $\widetilde{A}_{\boldsymbol{x}\boldsymbol{y}} \geq 0$ for $\boldsymbol{x} \neq \boldsymbol{y}$ (non-negative off-diagonals),
\item $\widetilde{A}_{\boldsymbol{x}\boldsymbol{x}} \leq 0$ for all $\boldsymbol{x}$ (non-positive diagonal),
\item $e^{t\widetilde{\mathbf{A}}_{JJ}} \geq 0$ and $\|e^{t\widetilde{\mathbf{A}}_{JJ}}\|_{1 \to 1} = 1$ for all $t \geq 0$ (generates a positive stochastic semigroup).
\end{enumerate}
\end{proposition}

\begin{proof}
Properties (1)-(3) follow directly from the construction of $\widetilde{\mathbf{A}}_{JJ}$. For property (1), using Eq.~\eqref{eq:row-sum-identities} and $\mathbf{1}^\top \diag(\boldsymbol{v}) = \boldsymbol{v}^\top$,
\begin{align*}
\mathbf{1}^\top \widetilde{\mathbf{A}}_{JJ} = \mathbf{1}^\top \mathbf{A}_{JJ} + \mathbf{1}^\top \mathbf{A}_{J'J} = \mathbf{0}^\top.
\end{align*}
Property (2) holds because off-diagonal entries of $\widetilde{\mathbf{A}}_{JJ}$ equal those of $\mathbf{A}_{JJ}$, which are non-negative. For property (3), the diagonal entries are $\widetilde{A}_{xx} = A_{xx} + \sum_{y \in J'} A_{yx} = - \sum_{y \in J, y \neq x} A_{yx} \leq 0$. For property (4), we first establish the norm bound, then positivity. For any $\boldsymbol{p}(0) \geq \mathbf{0}$ with $\|\boldsymbol{p}(0)\|_1 = 1$, the solution $\boldsymbol{p}(t) = e^{t\widetilde{\mathbf{A}}_{JJ}}\boldsymbol{p}(0)$ satisfies
\begin{equation}
\frac{d}{dt}\|\boldsymbol{p}(t)\|_1 = \mathbf{1}^\top \frac{d\boldsymbol{p}}{dt} = \mathbf{1}^\top \widetilde{\mathbf{A}}_{JJ}\boldsymbol{p}(t) = (\mathbf{1}^\top \widetilde{\mathbf{A}}_{JJ})\boldsymbol{p}(t) = 0,
\end{equation}
by property (1). Therefore, $\|\boldsymbol{p}(t)\|_1 = \|\boldsymbol{p}(0)\|_1 = 1$ for all $t \geq 0$. Taking the supremum over all initial conditions yields $\|e^{t\widetilde{\mathbf{A}}_{JJ}}\|_{1 \to 1} = 1$. Next, we prove positivity: by properties (2)-(3), choose $\alpha \geq \max_{\boldsymbol{x}} |\widetilde{A}_{\boldsymbol{x}\boldsymbol{x}}|$ large enough that $\mathbf{B} := \widetilde{\mathbf{A}}_{JJ} + \alpha I$ has all non-negative entries. Then
\begin{align*}
e^{t\widetilde{\mathbf{A}}_{JJ}} &= e^{t(\mathbf{B} - \alpha I)} \\
                        &= e^{-\alpha t} e^{t\mathbf{B}} \\
                        &= e^{-\alpha t} \lim_{n \to \infty} \left(I + \frac{t}{n}\mathbf{B}\right)^n.
\end{align*}

Since $\mathbf{B} \geq 0$, each factor $I + \frac{t}{n}\mathbf{B} \geq 0$ for $t \geq 0$ and $n$ sufficiently large. Non-negative matrices are closed under multiplication, so $(I + \frac{t}{n}\mathbf{B})^n \geq 0$ for all $n$. Taking the limit and multiplying by the positive scalar $e^{-\alpha t} > 0$, we can conclude that $e^{t\widetilde{\mathbf{A}}_{JJ}} \geq 0$
\end{proof}
We now establish the contraction property for the truncated generator $\mathbf{A}_{JJ}$, which is needed for the local and global error bounds.

\begin{corollary}[Truncated generator contraction]
\label{cor:truncated-contraction}
The truncated generator $\mathbf{A}_{JJ}$ with non-positive column sums generates a positive contraction semigroup satisfying $e^{t\mathbf{A}_{JJ}} \geq 0$ and $\|e^{t\mathbf{A}_{JJ}}\|_{1 \to 1} \le 1$ for all $t \ge 0$.
\end{corollary}

\begin{proof}
For the norm bound, consider any $\boldsymbol{p}(0) \geq \mathbf{0}$ with $\|\boldsymbol{p}(0)\|_1 = 1$. Since $\mathbf{1}^\top \mathbf{A}_{JJ} \leq \mathbf{0}^\top$, we have
\begin{equation}
\frac{d}{dt}\|\boldsymbol{p}(t)\|_1 = \mathbf{1}^\top \mathbf{A}_{JJ}\boldsymbol{p}(t) = (\mathbf{1}^\top \mathbf{A}_{JJ})\boldsymbol{p}(t) \leq 0,
\end{equation}
which gives $\|e^{t\mathbf{A}_{JJ}}\|_{1 \to 1} \leq 1$. Positivity follows by the same shifting argument as in Proposition~\ref{prop:compressed-generator}.
\end{proof}

\begin{remark}[Infinite-dimensional extension]
\label{rem:infinite-dimensional}
The arguments in Proposition~\ref{prop:compressed-generator} and Corollary~\ref{cor:truncated-contraction} extend to the full CME on infinite state space $\mathcal{X}$ via the Hille-Yosida theorem~\cite{Ethier1986,pazy_1983}. In the infinite-dimensional setting, the operator $\mathbf{A}$ is unbounded on $\ell^1(\mathcal{X})$, requiring verification that $\mathbf{A}$ is closed with dense domain $D(\mathbf{A})$ and that the resolvent bounds $\|(\lambda I - \mathbf{A})^{-1}\|_{1 \to 1} \leq 1/\lambda$ hold for all $\lambda > 0$. These conditions guarantee that $\mathbf{A}$ generates a strongly continuous contraction semigroup. In our finite-dimensional truncation to $J$, all operators are automatically bounded and closed with domain equal to the entire space.
\end{remark}

\subsection{Local Error}

We decompose the local error into model error from state space truncation and time-stepping error from numerical approximation.

\subsubsection{Model Error from State Space Truncation}

Consider the dynamics on $J$ starting from the true state $\boldsymbol{p}_J(t_n)$. The true restricted dynamics satisfy

\begin{equation}
\label{eq:true-restricted}
\frac{d}{dt} \boldsymbol{p}_J(t) = \mathbf{A}_{JJ}\boldsymbol{p}_J(t) + \mathbf{A}_{JJ'}\boldsymbol{p}_{J'}(t),
\end{equation}
while the compressed dynamics evolve according to
\begin{equation}
\label{eq:compressed-dynamics-exact}
\frac{d}{dt} \widehat{\boldsymbol{p}}_J(t) = \widetilde{\mathbf{A}}_{JJ}\widehat{\boldsymbol{p}}_J(t).
\end{equation}
Define the local model error $\boldsymbol{\varepsilon}(t) := \boldsymbol{p}_J(t) - \widehat{\boldsymbol{p}}_J(t)$ with $\boldsymbol{\varepsilon}(t_n) = \mathbf{0}$.

\begin{theorem}[Local model error bound]
\label{thm:local-model-error}
Over a single time step $[t_n, t_{n+1}]$, the local model error satisfies
\begin{equation}
\|\boldsymbol{\varepsilon}(t_{n+1})\|_1 \leq \int_{t_n}^{t_{n+1}} (\Phi_{\mathrm{in}}(s) + \Phi_{\mathrm{out}}(s)) \; ds.
\end{equation}
where $\Phi_{\mathrm{in}}$ and $\Phi_{\mathrm{out}}$ are boundary fluxes entering and leaving the truncation boundary respectively (Eq.~\eqref{eq:flux-defs}).
\end{theorem}

\begin{proof}
Subtracting Eq.~\eqref{eq:compressed-dynamics-exact} from Eq.~\eqref{eq:true-restricted} gives the following inhomogeneous equation,
\begin{align*}
\frac{d}{dt}\boldsymbol{\varepsilon}(t) &= \mathbf{A}_{JJ}\boldsymbol{\varepsilon}(t) + \underbrace{\mathbf{A}_{JJ'}\boldsymbol{p}_{J'}(t) - \diag(\mathbf{1}^\top\mathbf{A}_{J'J})\widehat{\boldsymbol{p}}_J(t)}_{=: \boldsymbol{f}(t)}.
\end{align*}
The solution of the inhomogeneous ODE is given by,
\begin{align}
\boldsymbol{\varepsilon}(t_{n+1}) &= \int_{t_n}^{t_{n+1}} e^{(t_{n+1}-s)\mathbf{A}_{JJ}}\boldsymbol{f}(s) \, ds + e^{(t_{n+1}-t_n)\mathbf{A}_{JJ}}\boldsymbol{\varepsilon}(t_n). \\
\end{align}
We set the initial condition $\boldsymbol{\varepsilon}(t_n)$ to zero as we start both systems with the same initial condition. Eliminating the initial condition and taking the $\ell_1$ norm both sides, we get:
\begin{align}
\|\boldsymbol{\varepsilon}(t_{n+1})\|_1 &\leq \int_{t_n}^{t_{n+1}} \underbrace{\|e^{(t_{n+1}-s)\mathbf{A}_{JJ}}\|_{1}}_{\leq 1 (\text{Corollary \ref{cor:truncated-contraction}})}\|\boldsymbol{f}(s)\|_1 \, ds \\
&\leq \int_{t_n}^{t_{n+1}} \left(\|\mathbf{A}_{JJ'}\boldsymbol{p}_{J'}(s)\|_1  +\|\diag(\mathbf{1}^\top\mathbf{A}_{J'J})\widehat{\boldsymbol{p}}_J(s)\|_1 \right)  \; ds \\
&= \int_{t_n}^{t_{n+1}} ( (\mathbf{1}^\top \mathbf{A}_{JJ'})\boldsymbol{p}_{J'}(s) + (\mathbf{1}^\top\mathbf{A}_{J'J})\widehat{\boldsymbol{p}}_J(s)) \;ds \\
&= \int_{t_n}^{t_{n+1}} (\Phi_{\mathrm{in}}(s) + \Phi_{\mathrm{out}}(s)) \; ds.
\end{align}

\end{proof}

While potentially conservative, this characterization provides a computationally tractable error estimate based on readily available flux quantities.

\subsubsection{Matrix Exponential Approximation Error}
\label{subsec:time_stepping_error_analysis}

In addition to state-space truncation, we approximate the matrix exponential $\exp(\widetilde{\mathbf{A}}_{JJ}\,\Delta t)$ numerically.

\begin{proposition}[Local time-stepping error]
\label{prop:local_time_error}
For a matrix exponential approximation with tolerance $\varepsilon_{\mathrm{ODE}}$, the error in a single time step is bounded by
\begin{equation}
\|\exp(\widetilde{\mathbf{A}}_{JJ}\,\Delta t)\,\boldsymbol{v} - \mathrm{Approx}(\widetilde{\mathbf{A}}_{JJ}, \Delta t)\,\boldsymbol{v}\|_1 \le \varepsilon_{\mathrm{ODE}}\,\|\boldsymbol{v}\|_1.
\end{equation}
Modern Krylov subspace implementations such as EXPOKIT \cite{Sidje1998} provide adaptive mechanisms to control this error by adjusting the Krylov subspace dimension based on the specified tolerance.
\end{proposition}

Combining Theorem~\ref{thm:local-model-error} with the pruning error gives the following complete one-step bound.  The tolerance $\alpha$ (probability mass threshold) appears directly; $\varepsilon_{\mathrm{flux}}$ does not appear directly but controls which states end up in $J'_n$, keeping $w_{\max}(J'_n)$ small in practice.

\begin{corollary}[Local step error]
\label{cor:local-error-pruning}
Under flux-based pruning at step $n$ with probability tolerance $\alpha$, the one-step error satisfies
\begin{equation}
\|\boldsymbol{\varepsilon}(t_{n+1})\|_1 \leq 2\alpha + \bigl(\alpha\,w_{\max}(J'_n) + \Phi_{\mathrm{out}}(J_n,t_n)\bigr)\Delta t_n + O(\Delta t_n^2),
\end{equation}
where $w_{\max}(J'_n) := \max_{\boldsymbol{x}\in J'_n} w(\boldsymbol{x})$ is the maximum exit rate of pruned states and $\Phi_{\mathrm{out}}(J_n,t_n) := \|\mathbf{A}_{J'_nJ_n}\widetilde{\boldsymbol{p}}_n\|_1$ is the boundary outflux.
\end{corollary}
\begin{proof}
\emph{Pruning and renormalization.}  Let $m \leq \alpha$ be the probability mass discarded by Stage~1 quantile pruning.  After flux-protection (Stage~2) and renormalization (Stage~3), the renormalized distribution $\widetilde{\boldsymbol{p}}_n$ satisfies
\[
\|\boldsymbol{p}(t_n) - \widetilde{\boldsymbol{p}}_n\|_1 = m + \frac{m}{1-m}(1-m) = 2m \leq 2\alpha,
\]
so $\|\boldsymbol{\varepsilon}(t_n)\|_1 \leq 2\alpha$.  The discarded mass obeys $\|\boldsymbol{p}_{J'}(t_n)\|_1 = m \leq \alpha$.

\emph{Evolution.}  Taking norms in the variation-of-constants formula from the proof of Theorem~\ref{thm:local-model-error} and using $\|e^{\tau\mathbf{A}_{JJ}}\|_{1\to1}\leq 1$ (Corollary~\ref{cor:truncated-contraction}),
\[
\|\boldsymbol{\varepsilon}(t_{n+1})\|_1 \leq \|\boldsymbol{\varepsilon}(t_n)\|_1 + \int_{t_n}^{t_{n+1}}\bigl(\Phi_{\mathrm{in}}(s)+\Phi_{\mathrm{out}}(s)\bigr)\,ds.
\]
Bounding each flux term at $t=t_n$ via the induced $\ell^1$ operator norm:
\begin{align*}
\Phi_{\mathrm{in}}(t_n) &= \|\mathbf{A}_{JJ'}\boldsymbol{p}_{J'}(t_n)\|_1 \leq \|\mathbf{A}_{JJ'}\|_{1\to1}\,\|\boldsymbol{p}_{J'}(t_n)\|_1 \leq w_{\max}(J'_n)\cdot\alpha, \\
\Phi_{\mathrm{out}}(t_n) &= \|\mathbf{A}_{J'J}\widetilde{\boldsymbol{p}}_n\|_1 =: \Phi_{\mathrm{out}}(J_n,t_n),
\end{align*}
where $\|\mathbf{A}_{JJ'}\|_{1\to1}\leq w_{\max}(J'_n)$ because each column of $\mathbf{A}_{JJ'}$, indexed by $\boldsymbol{x}\in J'_n$, sums to at most $w(\boldsymbol{x})$.  Substituting $\|\boldsymbol{\varepsilon}(t_n)\|_1\leq 2\alpha$ and integrating to leading order gives the result.
\end{proof}

The $2\alpha$ term captures the immediate effect of pruning and renormalization.  The influx contribution $\alpha\,w_{\max}(J'_n)$ reflects the fact that the pruned mass $\boldsymbol{p}_{J'}$ continues to drive probability current back into $J_n$ in the true dynamics.  The quantity $w_{\max}(J'_n)$ is available as a byproduct of Algorithm~\ref{alg:flux-prune}: propensities are evaluated for all active states \emph{before} any are removed, so the maximum exit rate over the pruned set requires no extra computation.  The outflux $\Phi_{\mathrm{out}}(J_n,t_n)$ is suppressed by the EXPAND step, which adds all stoichiometric neighbors of active states so that outgoing transitions are absorbed within $J_n$.

\subsection{Global Error}

Corollary~\ref{cor:local-error-pruning} bounds the error introduced at a single step.  The following corollary shows how these per-step errors accumulate over $N$ steps via the contractivity of the CME semigroup.

\begin{corollary}[Global error bound]
\label{cor:global-error}
Let $e_n := \|\boldsymbol{p}(t_n) - \widetilde{\boldsymbol{p}}_n\|_1$ and $\bar{\Phi}_{\mathrm{out}} := \max_n \Phi_{\mathrm{out}}(J_n,t_n)$.  For $N$ steps with $e_0=0$,
\begin{equation}
\label{eq:global-error-bound}
e_N \leq N\bigl(2\alpha + \bar{\Phi}_{\mathrm{out}}\,\Delta t + \varepsilon_{\mathrm{ODE}}\bigr),
\end{equation}
where $\Delta t$ denotes a representative step size.  After the EXPAND step, $\bar{\Phi}_{\mathrm{out}} \approx 0$, so the dominant cost per step is $2\alpha$.
\end{corollary}

\begin{proof}
Let $\boldsymbol{r}(t)$ be the exact CME solution starting from $\widetilde{\boldsymbol{p}}_n$.  Insert $\boldsymbol{r}(t_{n+1})$ as an intermediate and apply the triangle inequality:
\[
e_{n+1} = \|\boldsymbol{p}(t_{n+1}) - \widetilde{\boldsymbol{p}}_{n+1}\|_1
\leq \underbrace{\|\boldsymbol{p}(t_{n+1}) - \boldsymbol{r}(t_{n+1})\|_1}_{\text{(I) prior error propagated}}
+ \underbrace{\|\boldsymbol{r}(t_{n+1}) - \widetilde{\boldsymbol{p}}_{n+1}\|_1}_{\text{(II) new error at step }n}.
\]
\textit{Term (I).}  The CME semigroup is a contraction in $\ell^1$ (Proposition~\ref{prop:compressed-generator}), so prior errors do not amplify under exact evolution:
\[
\|\boldsymbol{p}(t_{n+1}) - \boldsymbol{r}(t_{n+1})\|_1 \leq \|\boldsymbol{p}(t_n) - \widetilde{\boldsymbol{p}}_n\|_1 = e_n.
\]
Crucially, this includes the effect of the influx from pruned mass ($\boldsymbol{p}_{J'} \neq \boldsymbol{0}$ in the true dynamics): the contractivity of the full CME semigroup ensures the total $\ell^1$ distance between $\boldsymbol{p}$ and $\boldsymbol{r}$ does not grow, so prior pruning errors accumulate \emph{additively}.

\textit{Term (II).}  Since $\boldsymbol{r}$ and $\widetilde{\boldsymbol{p}}$ start from the same $\widetilde{\boldsymbol{p}}_n$ (which is supported on $J_n$), the pruned mass in $J'_n$ is zero: $\boldsymbol{p}_{J'}(t_n) = \boldsymbol{0}$, so $\Phi_{\mathrm{in}}(t_n) = 0$.  Applying Theorem~\ref{thm:local-model-error} and Proposition~\ref{prop:local_time_error},
\[
\|\boldsymbol{r}(t_{n+1}) - \widetilde{\boldsymbol{p}}_{n+1}\|_1 \leq \underbrace{2\alpha}_{\text{prune at }t_{n+1}} + \underbrace{\Phi_{\mathrm{out}}(J_n,t_n)\,\Delta t_n + \varepsilon_{\mathrm{ODE}}}_{\text{truncation + ODE error}}.
\]
Combining: $e_{n+1} \leq e_n + 2\alpha + \Phi_{\mathrm{out}}(J_n,t_n)\,\Delta t_n + \varepsilon_{\mathrm{ODE}}$.  Taking maxima over $n$ and summing from $n=0$ to $N-1$ with $e_0=0$ gives the result.
\end{proof}

\section{Numerical Experiments}
\label{sec:results}

In this section we demonstrate the flux-aware FSP method on four benchmark systems from the stochastic chemical kinetics literature. We first present a detailed analysis on a toy bottleneck reaction system that validates the method's connectivity-preserving properties and error bounds along with a deterministic error analysis. The toggle switch tests spatial adaptivity for bimodal distributions. The Robertson system tests long-time integration with extreme temporal stiffness spanning nine orders of magnitude in rate constants. The Oregonator tests adaptive time-stepping for sustained oscillatory dynamics. These examples illustrate performance on well-established models with large state spaces and complex transient behavior. Finally, we experimentally study the matrix construction algorithm's efficiency. All experiments were conducted on an M1 MacBook Pro (M1 Pro chip, 16 GB RAM). The adaptive FSP algorithm was implemented in Julia, using Catalyst.jl~\cite{catalyst2021} for reaction network modeling and Expokit.jl for Krylov-based matrix exponential computations via EXPOKIT's routines~\cite{expokit2021}.

Tolerance parameters were set from the known rate constants of each system rather than by tuning.  For $\varepsilon_{\Delta t}$, well-mixed systems (Oregonator, Robertson, toggle switch) use $\varepsilon_{\Delta t} = 1$, since $\Phi_{\mathrm{total}}$ already scales with the dominant reaction rates and this value keeps the per-step activity fraction bounded.  For $\alpha$, connectivity-critical systems such as the bottleneck admit aggressive pruning ($\alpha = 0.9$) because the flux criterion independently protects any connector state; oscillatory and bimodal systems use moderate values ($\alpha = 0.1$--$0.3$) to avoid over-pruning the tails of a moving or split distribution.  For $\varepsilon_{\mathrm{flux}}$, simply-connected systems use $\varepsilon_{\mathrm{flux}} = 1$: since no individual state can carry flux equal to the total system flux $\Phi_{\mathrm{total}}$, the protection threshold $\varepsilon_{\mathrm{flux}}\cdot\Phi_{\mathrm{total}}$ is unreachable by any single state, which effectively disables flux-based protection and allows the quantile filter to operate freely.  For the bottleneck system, $\varepsilon_{\mathrm{flux}} = 10^{-12}$ was set to match the estimated relative flux of the connector state $p(B{=}1)\,w(B{=}1)/\Phi_{\mathrm{total}} \approx k_1/k_2 = 10^{-6}/0.1 = 10^{-5}$, and the results were insensitive to any value in $[10^{-12}, 10^{-6}]$.

\subsection{Reaction System With a Bottleneck Channel}

We consider a three-species network with a slow gateway followed by fast downstream dynamics:
\begin{equation}
A \xrightarrow{k_1} B, \qquad B \xrightarrow{k_2} B + C,
\end{equation}
with state $(N_A,N_B,N_C)\in\mathbb{Z}_{\ge 0}^3$, $k_1=10^{-6}\,\text{s}^{-1}$, $k_2=0.1\,\text{s}^{-1}$, and initial condition $(1,0,0)$. The reaction $B \to B+C$ conserves $B$ while producing $C$, creating a catalytic production pathway. The slow $A \to B$ transition acts as a bottleneck: once probability reaches states with $B \geq 1$, the fast catalytic reaction rapidly advances $C$ while maintaining $B \approx 1$.

\begin{figure}[htbp]
\centering
\includegraphics[width=\linewidth]{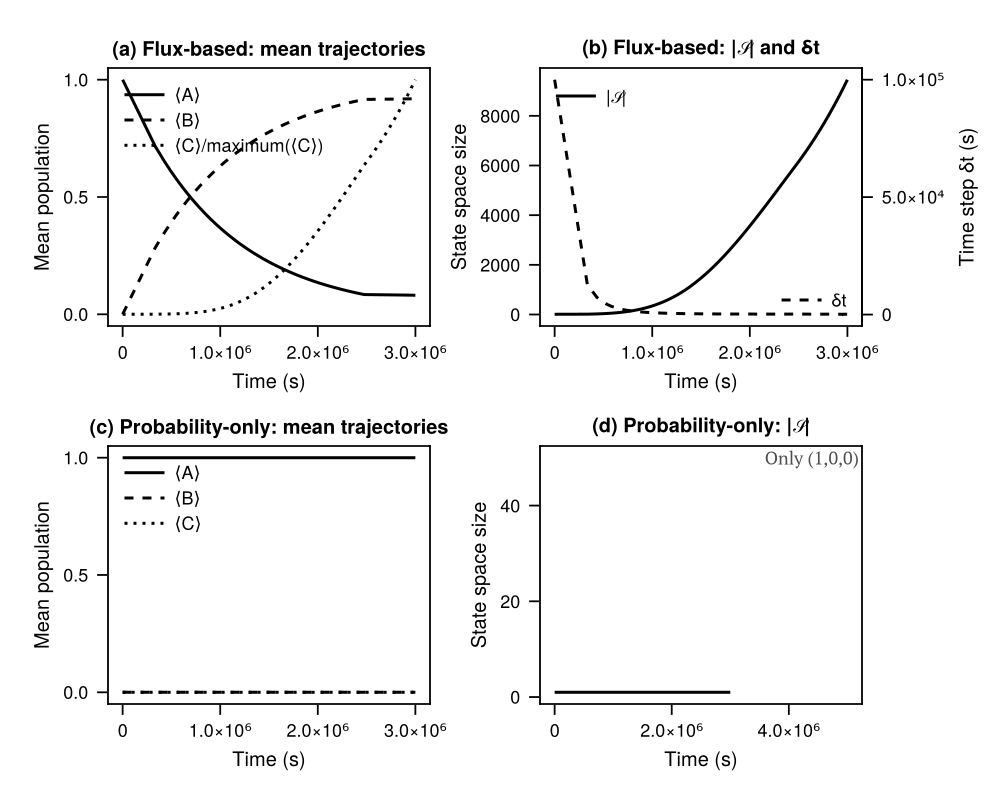}
\caption{Comparison of flux-based and probability-only pruning on the bottleneck system. (a) Flux-based pruning ($\alpha=0.9$, $\varepsilon_{\mathrm{flux}}=10^{-12}$) correctly captures the transition dynamics with $\langle A \rangle$ decaying and $\langle C \rangle$ growing. (b) State space grows continuously to accommodate the advancing wave front while time step $\delta t$ adjusts adaptively. (c) Probability-only pruning ($\alpha=0.9$, no flux preservation) removes all intermediate states immediately, severing the connectivity pathway and causing all means to remain at initial values. (d) Only the initial state $(1,0,0)$ is retained throughout the simulation, creating a reducible generator.}
\label{fig:flux_vs_naive}
\end{figure}

The probability mass concentrates at the initial state $(1,0,0)$ and a traveling wave front at $(0,1,C_{\max}(t))$ where $C_{\max}$ grows approximately linearly with time. Intermediate states $(0,1,C)$ for $0 < C < C_{\max}$ form a critical connectivity pathway with negligible probability but significant flux. Tables~\ref{tab:bottleneck_diagnostics} and \ref{tab:prob_flux_rankings} quantify this structure at representative timepoints during the transition phase.

\begin{table}[htbp]
\centering
\caption{Distribution structure showing a quasi-bimodal behavior and bottleneck pathway characteristics during the stiff phase. The initial state $(1,0,0)$ gradually releases probability mass to the advancing tail $(0,1,C_{\max})$ of the distribution. Intermediate states along the pathway from $C=0$ to $C=C_{\max}$ maintain nearly constant flux despite having probability $p \lesssim 10^{-5}$. The pathway flux $\phi$ is computed as the total flux through intermediate bottleneck states and approximately equals $k_1 \cdot p(1,0,0)$.}
\label{tab:bottleneck_diagnostics}
\begin{tabular}{lccccc}
\hline
\textbf{Time (s)} & \textbf{$p(1,0,0)$} & \textbf{$C_{\max}$} & \textbf{$p(0,1,C_{\max})$} & \textbf{$|\mathcal{S}|$} & \textbf{Flux (s$^{-1}$)} \\
\hline
$1.0 \times 10^5$ & 0.905 & 1   & 0.095 & 4     & $9.0 \times 10^{-7}$ \\
$2.5 \times 10^5$ & 0.779 & 11  & 0.221 & 8     & $7.8 \times 10^{-7}$ \\
$5.1 \times 10^5$ & 0.600 & 33  & 0.400 & 36    & $6.0 \times 10^{-7}$ \\
$7.7 \times 10^5$ & 0.458 & 76  & 0.542 & 79    & $4.6 \times 10^{-7}$ \\
$3.0 \times 10^6$ & 0.081 & 8344 & 0.910 & 9124  & $8.4 \times 10^{-8}$ \\
\hline
\end{tabular}
\end{table}

\begin{table}[htbp]
\centering
\caption{Probability vs. flux rankings at early ($t=0$) and late ($t=3.0\times10^6$ s) times, demonstrating the disconnect between probability and flux. At $t=0$, the system has just initialized with all mass at $(1,0,0)$; zero-probability states $(0,1,0)$, $(0,1,1)$, $(0,1,2)$ carry the highest flux as they form the initial reaction pathway. At $t=3.0\times10^6$ s, state $(0,1,9121)$ near the wave front has negligible probability ($8\times10^{-4}$) but flux exceeding $10^{-4}$ s$^{-1}$. This is 100-fold higher than probability alone would suggest. Probability-only pruning would incorrectly remove all states except the top probability carrier, severing network connectivity.}
\label{tab:prob_flux_rankings}
\begin{tabular}{cccc}
\hline
\multicolumn{4}{c}{\textbf{Early time: $t = 1.0 \times 10^3$ s}} \\
\hline
\textbf{Rank} & \textbf{By probability} & \textbf{By flux} & \textbf{Flux (s$^{-1}$)} \\
\hline
1 & $(1,0,0)$, $p=0.999$ & $(0,1,2)$, $\phi=6.3\times10^{-5}$ & $6.3\times10^{-5}$ \\
2 & $(0,1,2)$, $p=6.3 \times 10^{-4}$ & $(0,1,1)$, $\phi=3.7\times10^{-5}$ & $3.7\times10^{-5}$ \\
3 & $(0,1,1)$, $p=3.7 \times 10^{-4}$ & $(0,1,0)$, $\phi=1.0\times10^{-6}$ & $1.0\times10^{-6}$ \\
4 & $(0,1,0)$, $p=1.0 \times 10^{-5}$ & $(1,0,0)$, $\phi=1.0\times10^{-6}$ & $1.0\times10^{-6}$ \\
\hline
\multicolumn{4}{c}{\textbf{Late time: $t = 3.0 \times 10^6$ s}} \\
\hline
\textbf{Rank} & \textbf{By probability} & \textbf{By flux} & \textbf{Flux (s$^{-1}$)} \\
\hline
1 & $(0,1,9122)$, $p=0.910$ & $(0,1,9121)$, $\phi=8.4\times10^{-5}$ & $8.4\times10^{-5}$ \\
2 & $(1,0,0)$, $p=0.081$ & $(0,1,2697)$, $\phi=2.0\times10^{-5}$ & $2.0\times10^{-5}$ \\
3 & $(0,1,9121)$, $p=8\times10^{-4}$ & $(0,1,4298)$, $\phi=3.1\times10^{-6}$ & $3.1\times10^{-6}$ \\
4 & $(0,1,9120)$, $p=2\times10^{-4}$ & $(0,1,4755)$, $\phi=3.2\times10^{-7}$ & $3.2\times10^{-7}$ \\
\hline
\end{tabular}
\end{table}

At early times ($t \sim 10^5$ s), the system exhibits a clear quasi-bimodal distribution with significant mass at both $(1,0,0)$ and the advancing front. As time progresses, probability mass transfers from the trapped state to the wave front, with $p(1,0,0)$ decreasing from 0.905 to 0.081 while $p(0,1,C_{\max})$ increases from 0.095 to 0.910 over the observed time range. The state space grows continuously from 4 to over 9000 states to accommodate the advancing front, with $C_{\max}$ increasing approximately linearly at rate $k_2 \approx 0.1$ s$^{-1}$ once $B=1$ is established.

Table~\ref{tab:prob_flux_rankings} reveals the critical disconnect between probability and flux rankings. At $t=1.0\times10^3$ s (shown with flux tolerance $\epsilon=10^{-3}$ to capture early dynamics), while $(1,0,0)$ dominates by probability with $p=0.999$, state $(0,1,2)$ at the wave front carries the highest flux of $6.3\times10^{-5}$ s$^{-1}$ despite having probability $p=6.3\times10^{-4}$—a flux-to-probability ratio of 100. State $(0,1,1)$ exhibits a similar pattern with flux $3.7\times10^{-5}$ s$^{-1}$ and probability $3.7\times10^{-4}$, while the bottleneck connector $(0,1,0)$ with $p=1.0\times10^{-5}$ carries flux comparable to the dominant state. At $t=3.0\times10^6$ s, this disconnect becomes even more pronounced: state $(0,1,9121)$ ranks third by probability with $p=8\times10^{-4}$ but first by flux with $\phi=8.4\times10^{-5}$ s$^{-1}$. Similarly, states $(0,1,2697)$ and $(0,1,4298)$ appear nowhere in the top probability rankings yet carry flux values of $2.0\times10^{-5}$ and $3.1\times10^{-6}$ s$^{-1}$, respectively. These states form the critical bottleneck chain connecting the trapped region at $(1,0,0)$ to the active wave front at $(0,1,9122)$.

The bottleneck pathway flux $\phi \approx k_1 \cdot p(1,0,0)$ decreases proportionally as mass depletes from the initial state, ranging from $9.0 \times 10^{-7}$ s$^{-1}$ at early times down to $8.4 \times 10^{-8}$ s$^{-1}$ at $t = 3 \times 10^6$ s. Despite this decreasing flux magnitude, the intermediate states maintain flux values that are orders of magnitude larger than their probabilities, as quantified in Table~\ref{tab:prob_flux_rankings}.

Figure~\ref{fig:flux_vs_naive} compares flux-based pruning against probability-only pruning. The flux-based method correctly preserves intermediate states despite their negligible probability, maintaining network connectivity and capturing the transition dynamics. State space growth is adaptive and continuous, expanding to accommodate the advancing wave front. The adaptive time step $\delta t$ begins at $\sim 10^5$ s and decreases to $\sim 10^3$ s during rapid transitions before stabilizing.

In contrast, probability-only pruning removes all states except $(1,0,0)$ immediately, creating a reducible generator where the initial state becomes absorbing. This severs the connectivity pathway entirely, causing all mean values to remain at their initial conditions: $\langle A \rangle = 1$, $\langle B \rangle = 0$, $\langle C \rangle = 0$ for all time. This method incorrectly predicts zero production of species $C$ when the correct answer shows $\langle C \rangle$ growing linearly with time to reach thousands by $t = 10^6$ s.

This example demonstrates that bottleneck states are essential for preserving network ergodicity and accurate long-time dynamics. Flux-based pruning correctly identifies and retains these states even when they carry negligible mass, while probability-based thresholds incorrectly remove them. The method successfully maintains connectivity across state spaces ranging from single-digit to thousands of states, adapting the representation as the distribution evolves.

\paragraph{Expansion strategy.}
With only $R=2$ reactions, stoichiometric expansion adds at most two neighbors per active state per step, so state-space growth is naturally bounded.  We use \emph{stoichiometric expansion} throughout; the flux-adaptive step size naturally limits to $\sim 1/k_2 = 10$\,s during the bottleneck phase, keeping expansion pace-matched with the probability wavefront.

\subsubsection{Experimental Error Analysis}
\begin{figure}[htbp]
\centering
\includegraphics[width=\textwidth]{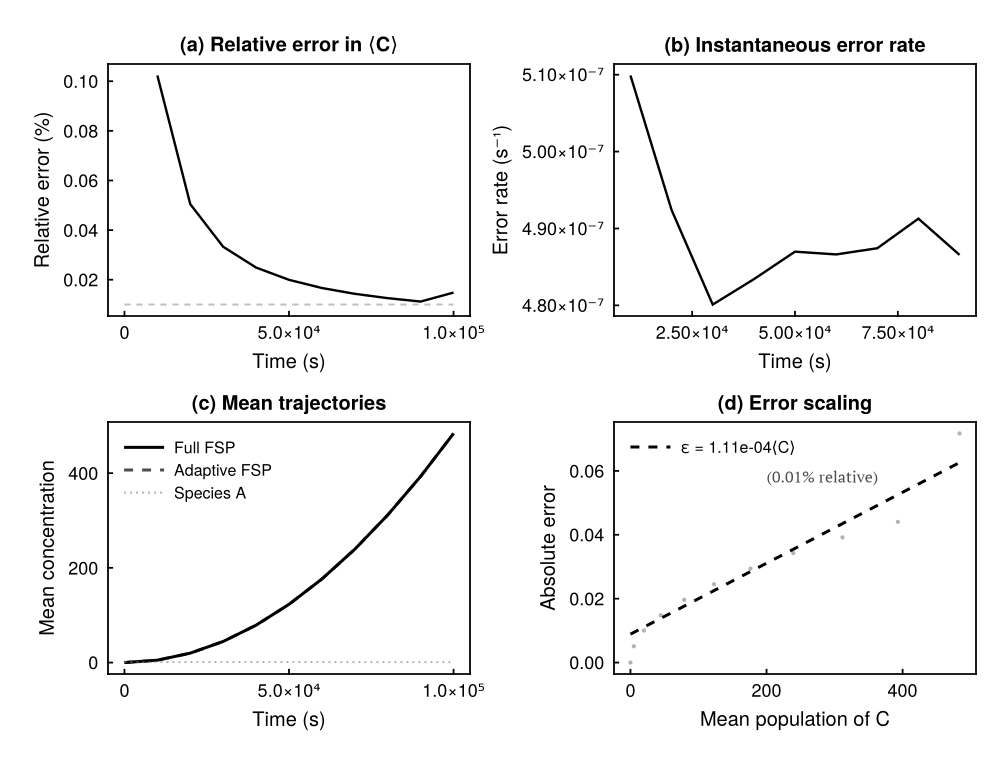}
\caption{Error accumulation analysis over $10^5$ seconds with flux tolerance $\varepsilon_{\mathrm{flux}} = 10^{-6}$. (a) Relative error in mean of species $C$ decays from $0.1\%$ to $0.015\%$ as the distribution spreads. (b) Instantaneous error rate remains approximately constant at $4.9 \times 10^{-7}$ s$^{-1}$ throughout the stable region. (c) Mean trajectories for species $C$ (full FSP vs.\ adaptive) and species $A$ (reference). (d) Absolute error scales linearly with mean population: $\epsilon \approx 1.11 \times 10^{-4} \langle C \rangle$, corresponding to $0.01\%$ relative error.}
\label{fig:error_accumulation}
\end{figure}

The slowly expanding state-space of the bottleneck system allows us to experimentally compare the flux-aware adaptive FSP against the full FSP method over extended time periods. Because the bottleneck dynamics confine probability mass to a manageable region, we can compute the full FSP solution as a reference and directly measure the accuracy of our adaptive approach. We validate Corollary~\ref{cor:global-error} by comparing the adaptive FSP (quantile + flux-preserving pruning) against an exact FSP on a fixed state space $\{0,\ldots,2\} \times \{0,\ldots,2\} \times \{0,\ldots,10000\}$. The exact solution evaluates the matrix exponential at times $t \in \{0, 10^4, 2\times10^4, \ldots, 10^5\}$ s, while the adaptive method uses $\delta t = 10$ s steps with dynamic truncation.

Figure~\ref{fig:error_accumulation} shows error behavior over the full time interval. Panel (a) demonstrates that relative error in $\langle C \rangle$ decreases from $0.1\%$ to $0.015\%$ as the distribution spreads—the absolute error grows with the mean, but represents a shrinking fraction of the population. Panel (b) reveals constant instantaneous error rate of approximately $4.9 \times 10^{-7}$ s$^{-1}$, confirming the additive error accumulation predicted by Corollary~\ref{cor:global-error}. Panel (c) shows visually indistinguishable mean trajectories for species $C$. Panel (d) establishes linear error scaling $\epsilon \approx 1.11 \times 10^{-4} \langle C \rangle$, yielding $0.01\%$ relative error maintained across three orders of magnitude in population.

\begin{table}[htbp]
\centering
\caption{Distribution errors at $t = 10^5$ s.}
\label{tab:distribution_errors}
\begin{tabular}{lccc}
\hline
\textbf{Metric} & \textbf{Species $A$} & \textbf{Species $B$} & \textbf{Species $C$} \\
\hline
Mean (Adaptive) & $0.9903$ & $9.71 \times 10^{-3}$ & $474.6$ \\
Mean (Exact) & $0.9900$ & $9.95 \times 10^{-3}$ & $474.7$ \\
Absolute Error & $2.5 \times 10^{-4}$ & $2.5 \times 10^{-4}$ & $0.072$ \\
Relative Error & $0.025\%$ & $2.5\%$ & $0.015\%$ \\
\hline
\end{tabular}
\end{table}

Table~\ref{tab:distribution_errors} quantifies final-time accuracy. Species $C$ achieves $0.015\%$ relative error with absolute error $0.072$—the linear scaling $\epsilon \propto \langle C \rangle$ ensures errors remain proportional to population size. The adaptive state space contains $10{,}003$ states versus the exact method's $45{,}009$ states ($4.5\times$ reduction), though the adaptive space grows linearly in time as $C$ accumulates. Flux-based pruning correctly preserves low-probability bottleneck states at small $C$ values that remain connectivity-critical even as the distribution bulk moves to higher values. Aggressive probability pruning would sever these pathways and destroy long-time accuracy.

The results validate the error bound as $\mathcal{O}(10^{-4})$ relative error in means across $10^4$ time steps, constant error rate $\sim 5 \times 10^{-7}$ s$^{-1}$, and additive rather than multiplicative accumulation due to generator contraction.

\subsection{Stochastic Toggle Switch Model}

\begin{figure}[!htbp]
    \centering
    \includegraphics[width=\linewidth]{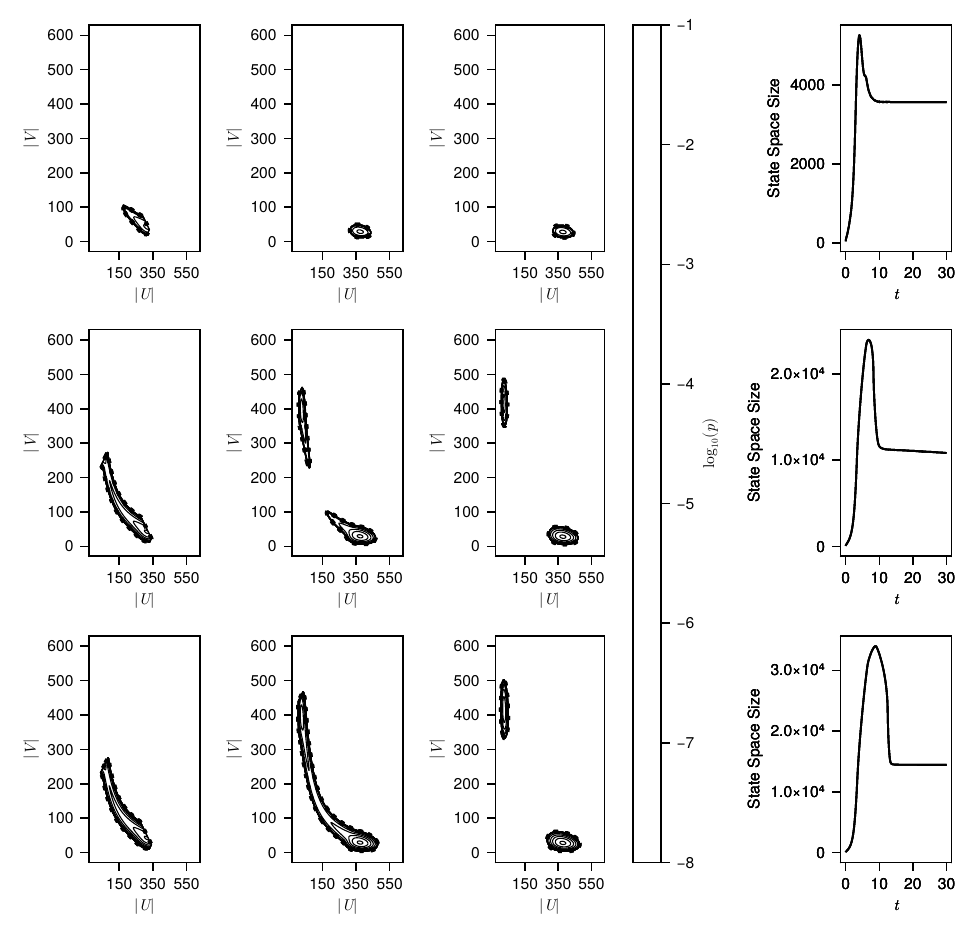}
    \caption{
    Joint probability contours of the stochastic toggle switch at three representative time points (\(t\approx 2, 15, 30\)). 
    Columns correspond to time, while rows correspond to different adaptive--FSP configurations:
    (top) pruning mass \(0.3\) with no flux filter;
    (middle) pruning mass \(0.3\) with flux tolerance \(10^{-7}\);
    (bottom) pruning mass \(0.001\) with no flux filter.
    Black contour lines indicate probability levels, and dotted curves indicate the truncation boundary.
    The right column shows the corresponding state-space size trajectories.
    }
    \label{fig:toggle_results}
\end{figure}

The stochastic toggle switch models a genetic regulatory network where two proteins ($U$ and $V$) mutually repress each other via Hill-type kinetics~\cite{tian_stochastic_2006}. The forward reactions incorporate cooperativity with Hill coefficient $n=3$ to model mutual repression, while the reverse reactions represent degradation. Table~\ref{table:toggle} details the reaction network.

\begin{table}[!htbp]
\centering
\begin{tabular}{|c|c|c|}
\hline
\textit{Reaction} & \textit{Reaction Equation} & \textit{Parameter Values} \\
\hline
1 & \(\emptyset \xrightarrow{\eta\left(\alpha_1 + \frac{\beta_1K_1^3}{K_1^3+V^3}\right)} U\) & $\eta = 1.0$, $\alpha_1 = 20$, $\beta_1 = 400$, $K_1 = 100$ \\
2 & \(U \xrightarrow{d_1+\frac{s\gamma}{1+s}} \emptyset\) & $d_1 = 1.0$, $s = 0.1$, $\gamma = 1.0$ \\
3 & \(\emptyset \xrightarrow{\eta\left(\alpha_2 + \frac{\beta_2K_2^3}{K_2^3+U^3}\right)} V\) & $\eta = 1.0$, $\alpha_2 = 20$, $\beta_2 = 400$, $K_2 = 100$ \\
4 & \(V \xrightarrow{d_2} \emptyset\) & $d_2 = 1.0$ \\
\hline
\end{tabular}
\caption{Reactions and parameter values for the stochastic toggle switch model~\cite{tian_stochastic_2006, Dinh_2020}. Forward reactions incorporate Hill-type kinetics with cooperativity coefficient $n=3$ to capture mutual repression; reverse reactions represent degradation.}
\label{table:toggle}
\end{table}

The system exhibits bistability as probability mass accumulates in two spatially separated regions corresponding to high expression of one protein and low expression of the other. Noise-driven transitions between these stable states occur over long timescales, requiring the adaptive method to simultaneously maintain multiple disconnected regions of state space while tracking rare switching events.

We initialize at $(U, V) = (85, 5)$ with reaction rate parameters taken from~\cite{Dinh_2020} and volume scaling factor $\eta = 100$. The simulation evolves over $t \in [0, 30]$ with initial time step $\delta t = 0.05$. This setup produces a bimodal distribution by the final time, with stable peaks near $(U, V) \approx (90, 10)$ and $(10, 90)$.

Figure~\ref{fig:toggle_results} examines the evolution of the joint distribution under three variants of the adaptive FSP scheme.  
In all cases, the system is initialized at \((U,V)=(85,5)\) and simulated to \(t=30\).  
Each column of the figure shows contour lines of \(\log_{10}(p)\) at three representative times, and each row corresponds to one truncation scheme.

With pruning mass \(0.3\) and no flux filter (top row), the state space expands rapidly after initialization, resolving the dominant drift into the lower metastable well.  
However, the truncation boundary is somewhat irregular and tends to either lag behind or overshoot the motion of the probability mass.  
This behavior reflects the fact that pruning is governed solely by a coarse quantile threshold as many intermediate transition states fall below the quantile cutoff and may be removed prematurely, while other regions are retained longer than necessary.

Introducing a flux tolerance of \(10^{-7}\) while keeping the pruning mass at \(0.3\) produces a dramatically more stable truncation (middle row).  
Here the boundary follows the outermost probability contours much more tightly, especially during the intermediate regime (\(t\approx 10\)--\(20\)) when the distribution stretches along the switching manifold connecting the two wells.  
The flux constraint protects states that carry even a small outward probability flow, preventing the removal of structurally important low-probability connectors.  
The resulting state space is noticeably smaller than in the quantile-only case, while still preserving the correct geometric structure of the distribution.

Reducing the pruning mass to \(0.001\) without a flux filter (bottom row) results in the most conservative truncation.  
The boundary expands to include a wide region of low-probability states, producing the largest active state space among the three experiments.  
Although this approach avoids under-truncation entirely, it incurs a substantial computational cost and offers little improvement over the flux-filtered configuration. Across all three settings, the adaptive FSP framework correctly resolves the emergence and separation of the two metastable wells. The comparison demonstrates the role of the flux tolerance as it preserves accuracy in the thin transition region without requiring an excessively small pruning mass, leading to a substantially more efficient state-space representation.

\subsection{Oscillatory Dynamics in the Oregonator Model}

We examine the Oregonator model, a reduced representation of the Belousov-Zhabotinsky oscillating chemical reaction developed by Field and Noyes \cite{field_noyes}. This three-species system captures the essential autocatalytic and inhibitory feedback mechanisms that generate sustained chemical oscillations. The Oregonator serves as a stringent test for numerical methods due to its combination of multiple timescales, relaxation oscillations, and sensitivity to parameter variations. The model consists of five elementary reactions involving three intermediate species: $X$ (HBrO$_2$), $Y$ (Br$^-$), and $Z$ (Ce$^{4+}$). The complete reaction network is presented in Table \ref{table:oregonator}. 

The reaction network exhibits three distinct dynamical phases. The autocatalytic reaction (R3) drives explosive growth of $X$ when inhibitor $Y$ is depleted. The quadratic termination (R4) prevents unbounded growth at high $X$ concentrations. The recovery reaction (R5) slowly regenerates $Y$ from $Z$, resetting the cycle. Following Gillespie's methodology \cite{Gillespie1977}, we determine rate constants by analyzing the deterministic steady state. At equilibrium, the net flux for each species vanishes:
\begin{align}
\frac{dN_X}{dt} &= k_1 N_Y - k_2 N_X N_Y + k_3 N_X - 2k_4 N_X^2 = 0, \\
\frac{dN_Y}{dt} &= -k_1 N_Y - k_2 N_X N_Y + k_5 N_Z = 0, \\
\frac{dN_Z}{dt} &= k_3 N_X - k_5 N_Z = 0.
\end{align}
We specify the steady-state populations as $(N_X^*, N_Y^*, N_Z^*) = (500, 1000, 2000)$ molecules, placing the system in the oscillatory regime. From the third equation, $k_3 N_X^* = k_5 N_Z^*$, which implies $\mu_3 = \mu_5$, where $\mu_i = k_i N_i^*$ denotes the steady-state flux through reaction $i$. Defining the characteristic fluxes $\mu_1^* = k_1 N_Y^* = 2000$ s$^{-1}$ and $\mu_2^* = k_2 N_X^* N_Y^* = 50000$ s$^{-1}$, we obtain the rate constants summarized in Table~\ref{table:oregonator}.

\begin{table}[!htbp]
\centering
\caption{Oregonator reaction network with mass-action kinetics and rate constants determined by steady-state flux balance at $(N_X^*, N_Y^*, N_Z^*) = (500, 1000, 2000)$ molecules.}
\label{table:oregonator}
\begin{tabular}{|c|l|c|l|}
\hline
\textit{Rxn} & \textit{Reaction} & \textit{Rate Constant} & \textit{Propensity Function}\\
\hline
1 & $Y \xrightarrow{k_1} X$ & $k_1 = 2.0$ s$^{-1}$ & $k_1 N_Y$ \\
2 & $X + Y \xrightarrow{k_2} \varnothing$ & $k_2 = 0.1$ molec$^{-1}$s$^{-1}$ & $k_2 N_X N_Y$ \\
3 & $X \xrightarrow{k_3} 2X + Z$ & $k_3 = 104$ s$^{-1}$ & $k_3 N_X$ \\
4 & $2X \xrightarrow{k_4} \varnothing$ & $k_4 = 0.016$ molec$^{-1}$s$^{-1}$ & $k_4 \frac{N_X(N_X-1)}{2}$ \\
5 & $Z \xrightarrow{k_5} Y$ & $k_5 = 26$ s$^{-1}$ & $k_5 N_Z$ \\
\hline
\end{tabular}
\end{table}

These parameters yield a period of approximately 0.85 time units with timescale separation spanning three orders of magnitude, specifically, the autocatalytic amplification (R3) at $k_3 = 104$ s$^{-1}$ drives explosive bursts, recovery (R5) operates at intermediate rate $k_5 = 26$ s$^{-1}$, and initiation (R1) proceeds slowly at $k_1 = 2.0$ s$^{-1}$. At steady-state concentrations, the inhibitory consumption (R2) dominates with flux $\mu_2^* = 50000$ s$^{-1}$.

At the steady-state initial condition, the total flux is $\Phi_{\mathrm{total}} \approx 158{,}000\,\text{s}^{-1}$, and with $\varepsilon_{\Delta t} = 1$ depth $d=1$ stoichiometric expansion suffices.

\paragraph{Expansion strategy.}
The Oregonator's oscillatory dynamics create a problem-specific constraint on the expansion strategy.  Stoichiometric expansion adds all five stoichiometric neighbors per active state; with $R = 5$ reactions and quantile parameter $\alpha = 0.1$, each step retains $0.9 \times 5 = 4.5$ net new states per existing state on average, causing exponential growth of $|J|$ that renders the simulation infeasible within 200 iterations (Section~\ref{sec:comparison}).  We therefore use \emph{SSA expansion}: at each step, one reaction direction is sampled per active state proportionally to its propensity and the resulting neighbor added to $J$.  At the steady-state initial condition the three dominant reactions (R2, R3, R5) each carry approximately $32\%$ of the total propensity, so SSA expansion concentrates the state space on the dynamically relevant frontier while keeping $|J| \lesssim 1{,}000$.  This is a system-specific choice; other problems with fewer or more uniform reaction channels use stoichiometric expansion.

Figure \ref{fig:oregonator_fsp_results} presents results over four complete oscillation cycles. The adaptive time step (panel a) varies over two orders of magnitude from $\Delta t_{\min} \approx 0.005$ during autocatalytic bursts to $\Delta t_{\max} \approx 0.04$ during recovery, tracking the instantaneous system flux and local stiffness.

The state space size (panel b) demonstrates excellent compression efficiency. Based on the observed mean trajectory ranges in panel c, a rectangular bounding box would require approximately $4000 \times 7000 \times 11000 \approx 3 \times 10^{11}$ states. Yet the adaptive method maintains only 50-2000 active states throughout the simulation: $|S| \approx 50$ during the recovery phase when the distribution concentrates tightly, expanding to $|S| \approx 2000$ during autocatalytic bursts when Poisson fluctuations spread the distribution. This achieves compression ratio $\approx 10^8:1$ while maintaining rigorous flux-based error control.

The mean trajectories (panel c) exhibit characteristic relaxation oscillations. Species $X$ spikes from near-zero to 6000 molecules in $< 0.1$ time units, consistent with the fast autocatalytic timescale $1/k_3$. Species $Y$ depletes rapidly during bursts via reactions R1 and R2, then recovers slowly via R5. Species $Z$ accumulates to 10000 molecules during each burst and decays with characteristic time $1/k_5 \approx 0.04$ time units. The phase relationships reflect the feedback structure where $Y$ depletion enables $X$ growth, $X$ growth produces $Z$, and $Z$ accumulation regenerates $Y$.

The phase space projections (panels d-f) reveal the limit cycle attractor. The $X$-$Y$ plane (panel d) shows the relaxation oscillation "knee": slow recovery where $Y$ increases with $X$ near zero, followed by rapid burst when $X$ spikes and $Y$ depletes. The $X$-$Z$ plane (panel e) displays an elliptical cycle with $Z$ lagging $X$ due to production-decay dynamics. The $Y$-$Z$ plane (panel f) demonstrates phase lag between inhibitor depletion and product accumulation. Counter-clockwise circulation in all projections confirms stable periodicity despite stochastic dynamics.
\begin{figure}[!htbp]
    \centering
    \includegraphics[width=\linewidth]{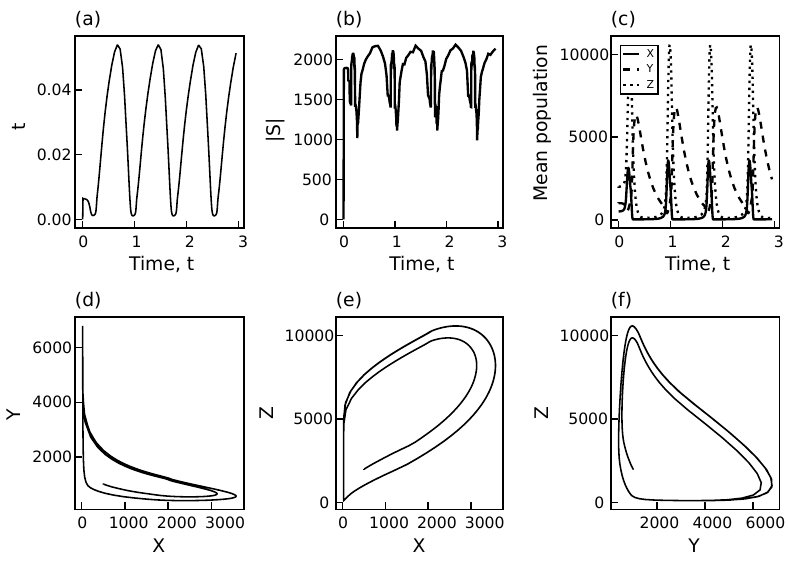}
    \caption{Adaptive FSP simulation of the Oregonator system. (a) Adaptive time step variation tracking instantaneous system stiffness. (b) Dynamic state space size showing compression ratio of $\sim 10^8:1$ relative to the theoretical state space of $3.6 \times 10^{11}$ states. (c) Mean population trajectories for species $X$ (solid), $Y$ (dash-dot), and $Z$ (dotted). (d-f) Phase space projections revealing the limit cycle attractor with intrinsic stochastic fluctuations.}
    \label{fig:oregonator_fsp_results}
\end{figure}

\subsection{Stiff Dynamics in the Robertson Autocatalytic System}

The Robertson system provides a severe test of numerical methods through its extreme stiffness, with reaction rates spanning nine orders of magnitude. Originally formulated by Robertson \cite{Robertson1967} for testing ODE solvers, this autocatalytic network has become a standard benchmark for stochastic simulation algorithms. The system involves three species undergoing the reactions shown in Table \ref{table:robertson}, with rate constants $k_1 = 0.04$ s$^{-1}$, $k_2 = 3 \times 10^7$ molecule$^{-1}$s$^{-1}$, and $k_3 = 10^4$ molecule$^{-1}$s$^{-1}$. This yields a stiffness ratio of $k_2/k_1 \approx 7.5 \times 10^8$, creating a temporal scale separation where the fast autocatalytic reaction R2 operates nearly a billion times faster than the slow initiation reaction R1.

\begin{table}[!htbp]
\centering
\caption{Robertson autocatalytic reaction system with extreme rate constant disparity}
\label{table:robertson}
\begin{tabular}{|c|l|c|l|}
\hline
\textit{Rxn} & \textit{Reaction} & \textit{Rate Constant} & \textit{Propensity Function} \\
\hline
1 & $A \xrightarrow{k_1} B$ & $k_1 = 0.04$ s$^{-1}$ & $a_1 = k_1 N_A$ \\
2 & $2B \xrightarrow{k_2} B + C$ & $k_2 = 3 \times 10^7$ molec$^{-1}$s$^{-1}$ & $a_2 = k_2 \frac{N_B(N_B-1)}{2}$ \\
3 & $B + C \xrightarrow{k_3} A + C$ & $k_3 = 10^4$ molec$^{-1}$s$^{-1}$ & $a_3 = k_3 N_B N_C$ \\
\hline
\end{tabular}
\end{table}

\begin{figure}[!htbp]
    \centering
    \includegraphics[width=\linewidth]{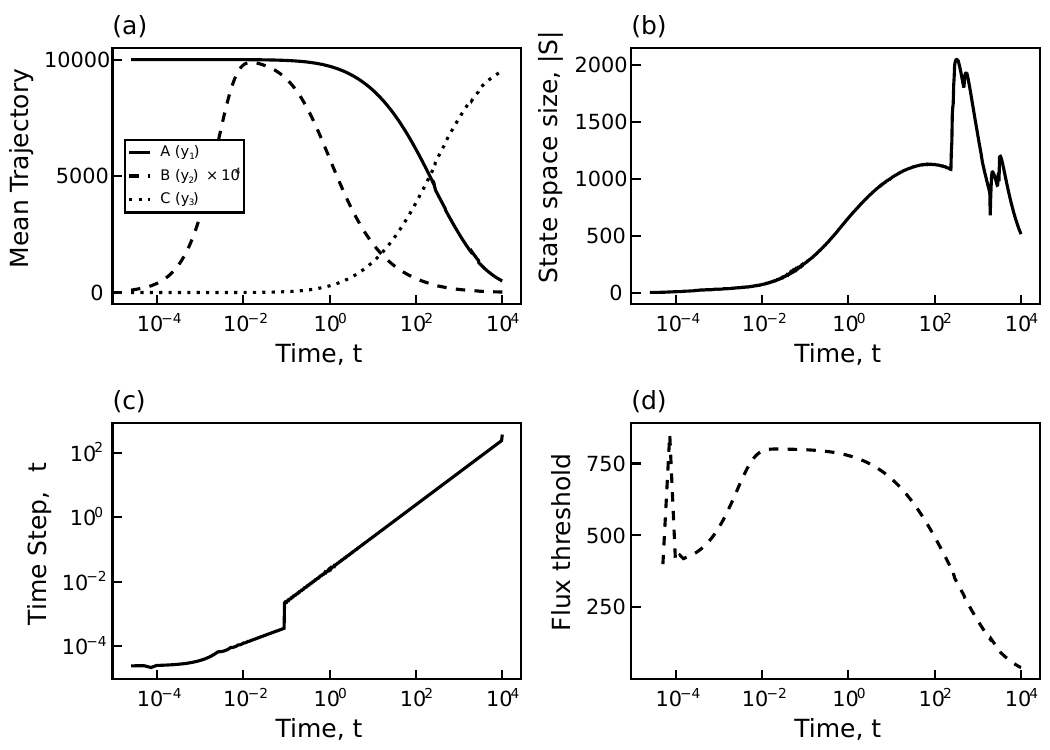}
    \caption{Adaptive FSP simulation of the Robertson system over eight decades of time ($10^{-5}$ to $10^5$ s). (a) Mean population trajectories for species $A$ (solid), $B$ (dashed), and $C$ (dotted) showing three-phase dynamics: slow initiation, explosive autocatalysis, and protracted equilibration. (b) Dynamic state space size achieving $\sim 10^4:1$ compression relative to the theoretical space of $3 \times 10^7$ states. (c) Adaptive time step spanning seven orders of magnitude in response to changing flux. (d) Total flux tracking system activity and guiding both time step and state space adaptation.}
    \label{fig:robertson_results}
\end{figure}

The system conserves total mass: $N_{\text{tot}} = N_A + N_B + N_C = 10000$ molecules. We initialize with $(N_A, N_B, N_C) = (10000, 0, 0)$, triggering characteristic three-phase dynamics that span eight decades of time. Figure~\ref{fig:robertson_results} demonstrates how the adaptive FSP method navigates this extreme stiffness through coordinated adjustment of time step, state space size, and flux monitoring.

At initialization, only R1 is active, giving $\Phi_{\mathrm{total}} = k_1 N_A = 400\,\text{s}^{-1}$.
As $B$ accumulates, R2 rapidly comes to dominate $\Phi_{\mathrm{total}}$ throughout the simulation.
With $\varepsilon_{\Delta t} = 0.01$ and $R = 3$ reactions, depth $d=1$ stoichiometric expansion suffices.

\paragraph{Expansion strategy.}
With $R=3$ reactions, stoichiometric expansion adds at most three neighbors per active state, keeping state-space growth tractable.  We use \emph{stoichiometric expansion}; this drives $\Phi_{\mathrm{out}} \to 0$ at each step, so the time-step bound $\|\boldsymbol{\varepsilon}\|_1 \lesssim \varepsilon_{\Delta t} = 0.01$ holds tightly and the step size is controlled entirely by the flux magnitude.

During the initial phase ($t < 10^{-2}$ s), species $A$ converts slowly to $B$ through reaction R1 at rate 0.04 s$^{-1}$. Panel (a) shows $A$ decreasing from 10000 molecules while $B$ (scaled $\times 10^4$) accumulates slowly and $C$ remains near zero. Time steps (panel c) remain near $\Delta t \approx 10^{-4}$ s, matching the dominant flux timescale. The state space (panel b) contains fewer than 50 states, reflecting the narrow distribution around the nearly deterministic trajectory. The flux threshold (panel d) hovers around 400 s$^{-1}$, dominated by reaction R1 with propensity $k_1 N_A \approx 400$ s$^{-1}$. This continues until $B$ reaches approximately 0.3 molecules near $t \approx 10^{-2}$ s.

The system then enters explosive transient phase ($10^{-2}$ s $< t < 10^2$ s) driven by reaction R2. Time steps contract to resolve the explosive dynamics, while the state space expands to over 2000 states as autocatalytic bursts amplify stochastic fluctuations scaling as $\sqrt{N_B}/N_B$. The flux threshold peaks near 800 s$^{-1}$. Using conservation law $N_A + N_B + N_C = 10000$ with observed ranges $N_A \in [0, 10000]$ and $N_B \in [0, 10]$, a rectangular bounding box requires approximately $10^5$ states. The adaptive method maintains only 2000 states at peak, achieving compression ratio $50:1$.

Following the transient ($t > 10^2$ s), reaction R3 slowly regenerates $A$ from $B$ and $C$, driving toward quasi-equilibrium near $(N_A, N_B, N_C) \approx (10000, 0, 0)$. Time steps increase progressively to $\Delta t \approx 10^2$ s by $t = 10^5$ s, spanning six orders of magnitude and the state space contracts to 500-1000 states as the distribution concentrates near quasi-equilibrium. The flux threshold decreases to approximately 100 s$^{-1}$ by $t = 10^4$ s, confirming approach to steady state.

\subsection{Performance Evaluation of Matrix Construction}

In this subsection we evaluate the matrix construction strategy introduced in Section~\ref{sec:matrix_construction}. The forward enumeration approach constructs the generator by iterating through reactions for each state, directly computing destination states as $y = x + \nu_k$. We benchmark this strategy against a standard implementation that uses backward state expansion on the Oregonator, ROBER, and toggle switch models.

Figure~\ref{fig:benchmark} and Table~\ref{tab:benchmark} quantify the performance difference on the ROBER model using $10^4$ trials via BenchmarkTools.jl~\cite{BenchmarkTools.jl-2016}. The benchmark captures a representative time step at $t \approx 170$ where the state space expands from 1106 to 1431 states. Standard construction exhibits median execution time of 6.95~ms with moderate variability (std. dev. 2.33~ms), while forward enumeration achieves 0.54~ms median time (std. dev. 1.12~ms), a 12.9$\times$ speedup. The cumulative distribution functions reveal that forward enumeration completes $>99\%$ of trials by 3~ms, whereas standard construction requires $>10$~ms for comparable completion rates.

\begin{figure}[!htbp]
\centering
\includegraphics[width=\textwidth]{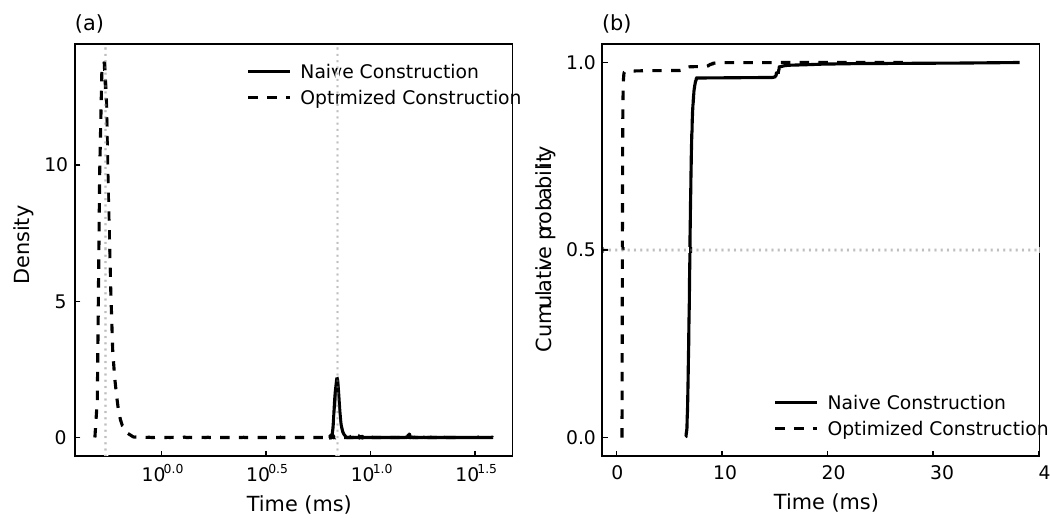}
\caption{Performance comparison of standard construction (solid line) and forward enumeration (dashed line) at $t \approx 170$ with state space sizes 1106 $\to$ 1431. (a) Probability density functions on logarithmic scale with median markers at 6.95~ms and 0.54~ms. (b) Cumulative distribution functions showing 12.9$\times$ speedup: forward enumeration completes 50\% of trials by 0.54~ms versus 6.95~ms for standard construction.}
\label{fig:benchmark}
\end{figure}

\begin{table}[h]
\centering
\caption{Benchmark statistics for matrix construction methods on ROBER model at $t \approx 170$ (state space sizes 1106 $\to$ 1431).}
\label{tab:benchmark}
\begin{tabular}{lccc}
\hline
Method & Median (ms) & Mean (ms) & Std. Dev. (ms) \\
\hline
Standard construction & 6.95 & 7.38 & 2.33 \\
Forward enumeration & 0.54 & 0.71 & 1.12 \\
\hline
Speedup & 12.9$\times$ & 10.4$\times$ & --- \\
\hline
\end{tabular}
\end{table}
Figure~\ref{fig:speed_comparison} demonstrates cumulative impact across three benchmark systems. The Oregonator (panel a) exhibits oscillatory state space dynamics with cost spikes exceeding $10^{-2}$~s during rapid expansion phases under standard construction; forward enumeration maintains cost below $10^{-3}$~s throughout. The stiff ROBER system (panel b) shows pronounced cost fluctuations between $10^{-3}$~s and $10^{-1}$~s with standard construction, reduced to approximately $10^{-4}$~s with forward enumeration. The toggle switch (panel c) demonstrates the largest proportional gain, with forward enumeration reducing per-step cost from $\sim$$10^{0}$~s to below $10^{-1}$~s.

\begin{figure}[!htbp]
\centering
\includegraphics[width=\textwidth]{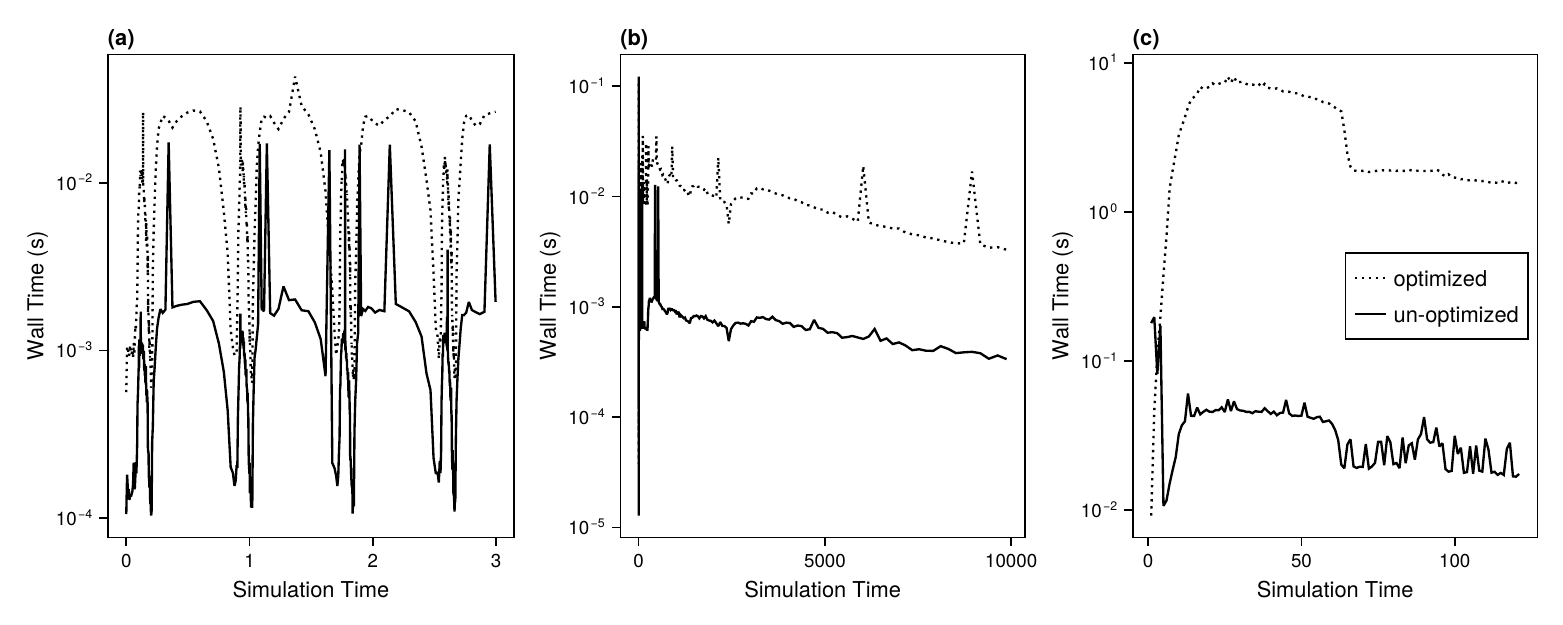}
\caption{Wall-clock time per integration step for FSP simulations with forward enumeration (solid) and standard construction (dotted). (a) Oregonator. (b) ROBER. (c) Toggle switch. Forward enumeration consistently outperforms standard construction across all systems.}
\label{fig:speed_comparison}
\end{figure}

Table~\ref{tab:walltime} summarizes total simulation times. The toggle switch achieves the largest speedup at 9.8$\times$ (364~s reduced to 37~s), benefiting from slower reaction timescales and compact state space geometry. The Oregonator and ROBER systems achieve 3.9$\times$ and 3.5$\times$ speedups respectively (4000~s $\to$ 1093~s and 10421~s $\to$ 3015~s). The consistent performance gains across diverse system types demonstrate that forward enumeration is the preferred construction method.

\begin{table}[h]
\centering
\caption{Total wall-clock times for complete FSP simulations.}
\label{tab:walltime}
\begin{tabular}{lccc}
\hline
System & Standard (s) & Optimized (s) & Speedup \\
\hline
Oregonator & 4000 & 1093 & 3.9$\times$ \\
ROBER & 10421 & 3015 & 3.5$\times$ \\
Toggle switch & 364 & 37 & 9.8$\times$ \\
\hline
\end{tabular}
\end{table}

\subsection{Comparison with Krylov-FSP-SSA}
\label{sec:comparison}

We compare with the Krylov-FSP-SSA algorithm of Vo and Sidje~\cite{Sidje2015}. This baseline shares the same expand--evolve--prune skeleton and Krylov matrix exponential integrator but differs in three key respects: (i) pruning uses a joint probability-and-derivative threshold (Algorithm~2, Step~7 of~\cite{Vo2017}: state $i$ is dropped only if $p_i < \texttt{droptol}$ \emph{and} $(\mathbf{A}\boldsymbol{p})_i < \texttt{droptol}'$, without flux-based protection), (ii) time stepping uses a mass-conservation accept/reject criterion (halving $\tau$ when retained mass falls below $1 - \varepsilon_{\mathrm{FSP}}(t+\tau)/t_f$) rather than the flux-adaptive rule $\Delta t = \varepsilon_{\Delta t}/\Phi$, and (iii) expansion is SSA-driven with $r$-step reachability enumeration. Figure~\ref{fig:unified_comparison} compares the two methods on the three benchmark systems; each exposes a qualitatively distinct failure mode. Krylov-FSP-SSA was run with $\varepsilon{=}10^{-2}$ and $r{=}1$ throughout, with $\tau_0{=}10^{-6}$ (Oregonator), $\tau_0{=}100$ (Bottleneck), and $\tau_0{=}0.1$ with the absorbing-form generator (Robertson). FP-FSP parameters follow the tolerance guidance of Section~\ref{sec:results}: $(\varepsilon_{\Delta t}{=}1,\,\alpha{=}0.1,\,\varepsilon_{\mathrm{flux}}{=}1)$ for the Oregonator; $(\delta t{=}10,\,\alpha{=}0.9,\,\varepsilon_{\mathrm{flux}}{=}10^{-12})$ for the Bottleneck; and $(\varepsilon_{\Delta t}{=}0.01,\,\alpha{=}0.4,\,\varepsilon_{\mathrm{flux}}{=}10^{-9})$ for Robertson.

\subsubsection{Oregonator.}
FP-FSP completes $t_f = 0.11$ in 36\,s with max $|\mathcal{S}| = 1{,}057$ (panel~(a)), while Krylov-FSP-SSA grows to 13{,}294 states and reaches only $t = 0.003$ (2.7\%) after 3{,}000 iterations and 74\,s. The explosion arises because the joint probability-and-derivative filter retains expanding boundary states with significant $(\mathbf{A}\boldsymbol{p})_i > \texttt{droptol}'$ even when $p_i$ is small; without flux-based pruning the active set grows monotonically. The $\tau$-lock at $\tau \approx 10^{-6}$ is fundamental: large propensities ($k_3 XZ \sim 10^5$) cause mass-conservation failures for any larger step, independent of $\varepsilon$, so parameter tuning cannot resolve it.

\subsubsection{Bottleneck.}
FP-FSP (panel~(b)) completes $t_f = 10^4$ in 1.0\,s with $\langle C \rangle = 4.978$ (${<}0.1\%$ error); the flux-adaptive step naturally limits to $\sim 1/k_2 = 10$\,s during the bottleneck phase, keeping depth-1 expansion pace-matched with the probability wavefront. Krylov-FSP-SSA ($\varepsilon=10^{-2}$, $\tau_{\mathrm{init}}=100$, $r=1$, $\mathtt{droptol}=10^{-10}$) also reaches $t_f$ in 7 accepted steps, but produces $\langle C \rangle \approx 0$ (${\gg}1{,}000\times$ underestimate). The failure is \emph{wavefront truncation}: the connector state $B{=}1$ is correctly retained ($p(B{=}1) \approx k_1 t \gg 10^{-10}$; both methods agree on $\langle B \rangle = 0.010$), but $\tau$ doubles after each accepted step (100\,$\to$\,3{,}700\,s) while $r=1$ adds only one new $C$-state per step. After 7 steps the projection covers only $C \leq 12$, while the true distribution of $C\mid B{=}1$ has spread over $[0,\,k_2 t_f] \approx [0, 1{,}000]$. With the absorbing-form generator, probability flux to $C \geq 13$ is absorbed into a sink; the retained mass within the 14-state projection is concentrated near $C{=}0$ and yields $\langle C \rangle \approx 0$. The mass-conservation criterion does not detect the inadequate coverage because the absorbed mass (${\approx}1\%$ by $t_f$) falls within the FSP tolerance. Tuning $(\varepsilon, \tau, r)$ over a wide range achieves at best $\langle C \rangle \approx 4.72$ (${\sim}5\%$ error) at comparable wall time to FP-FSP, because capping $\tau \leq 1/k_2$ forces ${\sim}1{,}000$ steps while recovering the efficiency advantage of large accepted steps requires $r \gtrsim k_2 \tau$, creating a fundamental conflict.

\subsubsection{Robertson.}
FP-FSP with the flux-adaptive step $\delta t = \varepsilon_{\Delta t}/\Phi$ tracks the non-monotonic burst-and-relax structure of the autocatalytic Robertson system (panel~(c)): $\delta t$ drops sharply during each burst as probability concentrates on high-flux states, then rises as the burst dissipates and $A$ depletes, peaking at $\approx 0.026$ during the equilibration phase ($t \gtrsim 3{,}000$). FP-FSP completes $t_f = 10^4$ in $662{,}472$ steps ($1{,}117$\,s) with $\langle A \rangle = 507$, $\langle C \rangle = 9{,}493$. Krylov-FSP-SSA with the correct absorbing-form generator (Algorithm~2 of~\cite{Vo2017}) incurs a 67\% rejection rate driven by FSP mass-conservation failures: the stiff $k_2 = 3{\times}10^7$ propensity causes the probability distribution to rapidly hit the projection boundary whenever $B{\geq}2$ states are populated, and the accepted step sizes $\tau$ fluctuate widely as the algorithm repeatedly halves $\tau$ to satisfy the FSP criterion. In 500 iterations ($1.9$\,s) it reaches only $t = 227.8$ ($2.3\%$ of $t_f$).

\begin{figure}[!htbp]
    \centering
    \includegraphics[width=\linewidth]{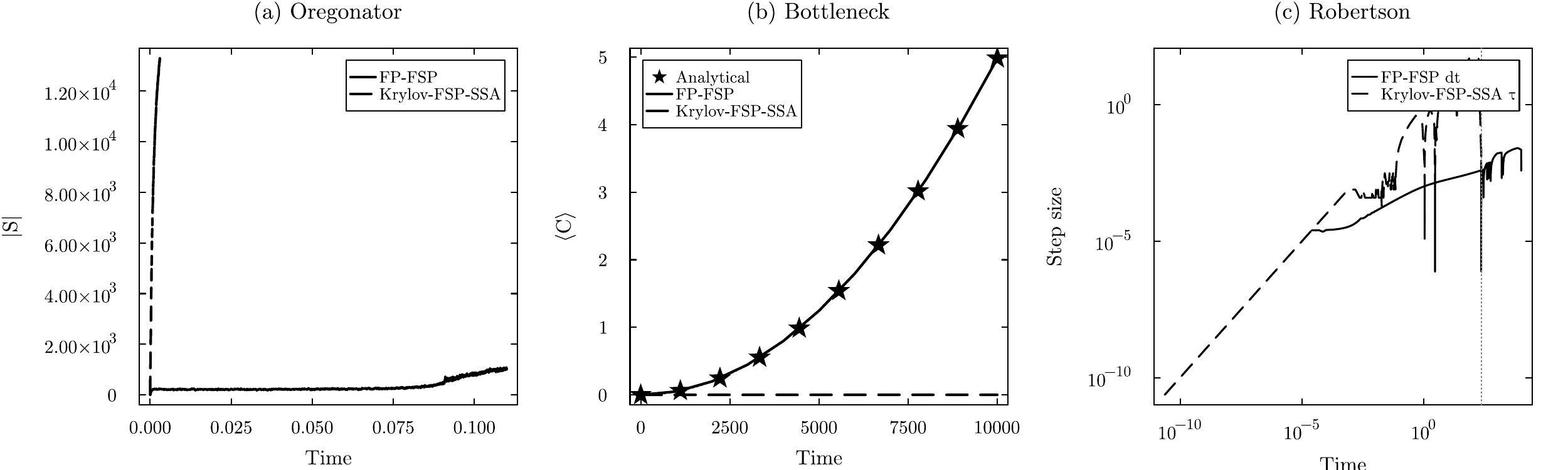}
    \caption{FP-FSP (solid) vs.\ Krylov-FSP-SSA (dashed) on three benchmark systems.
    (a)~Oregonator: state space size $|\mathcal{S}|$ vs.\ time. Krylov-FSP-SSA grows to ${\sim}13{,}300$ states; FP-FSP stays below 1{,}100 via flux-based pruning.
    (b)~Bottleneck ($A{\to}B{\to}B{+}C$, $t_f{=}10^4$): mean $\langle C \rangle$ vs.\ time with analytical reference (dotted, stars). FP-FSP tracks the exact solution; Krylov-FSP-SSA underestimates by $40{\times}$ due to wavefront truncation.
    (c)~Robertson ($t_f{=}10^4$): per-step $\delta t$ (FP-FSP, solid, full horizon) and accepted $\tau$ (Krylov-FSP-SSA, dashed; dotted vertical line marks $t{=}228$ where it stops) on a log scale. FP-FSP's non-monotonic $\delta t$ reflects repeated autocatalytic bursts; Krylov-FSP-SSA incurs 67\% rejections from $k_2{=}3{\times}10^7$ propensities and reaches only 2.3\% of $t_f$.}
    \label{fig:unified_comparison}
\end{figure}

\section{Discussion}
\label{sec:discussion}

The numerical results indicate that flux-based pruning can preserve network connectivity across a range of benchmark systems while maintaining competitive compression and accuracy. The comparison with the Krylov-FSP-SSA baseline in Section~\ref{sec:comparison} demonstrates that the two main contributions---flux-aware pruning and flux-adaptive time stepping---provide complementary robustness advantages on multiscale systems. The method targets two coupled difficulties in adaptive FSP: controlling the loss of important transitions when states are pruned, and selecting time steps that reflect the instantaneous activity level of the system. The flux-based pruning rule protects low-probability states that carry significant probability flow, and the flux-adaptive time-step selection aligns the numerical step size with the dominant rates in the truncated system.

The expansion rule for the proposed method is deliberately simple (fixed $r$-step expansion), and the pruning parameters $(\alpha, \varepsilon_{\mathrm{flux}})$ are fixed in time. This suffices for the test problems considered here, but more complicated networks may require additional structure in the adaptation strategy. For example, systems with strong directional drift or localized reaction zones could benefit from combining flux-based criteria with short trajectory bursts \cite{Dinh2016} to identify dynamically relevant regions before committing them to the FSP state space.

The present approach is particularly well suited to mass-conserving systems where the dynamically relevant region remains effectively bounded. Even in that setting, however, flux-aware pruning does not by itself prevent the active set from growing when probability mass is transported over long distances in state space. The bottleneck system in Section~\ref{sec:results} is an example: the distribution is quasi-bimodal, with mass slowly transferring from $(1,0,0)$ to $(0,1,C_{\max}(t))$ along a long chain of intermediate states that carry very low probability but high flux. Our flux-based criterion correctly retains this entire pathway for connectivity, but this also means many “bridge” states must remain active. In such situations, aggregation methods \cite{aggregation,peles,zhang_cao,Cao2016} and conservation-law-based reductions such as slack reactants \cite{slack} provide complementary tools as the intermediate bottleneck states are natural candidates to be lumped into a small number of macro-states while preserving the net flux across the bottleneck. Designing such hybrid schemes and maintaining a transparent error budget in the presence of aggregation remains an interesting problem for future work.

The choice of representation for the generator and the probability vector is another direction for extension. For the moderately sized state spaces in this work ($10^4$–$10^5$ states), Krylov-based exponential integrators \cite{Sidje1998} and sparse matrix representations are sufficient. For higher-dimensional systems, tensor-structured methods \cite{Kazeev2014,Dinh_2020} may be necessary. Extending flux diagnostics and pruning rules to tensor formats would require formulating flux quantities in terms of low-rank factors, which is nontrivial but may be essential for spatial or multicomponent models.

A more structural question is how the adaptive truncation perturbs the long-time behavior of the full chain. In this work we only use conductance and Cheeger-type ideas loosely to motivate flux preservation in Section~\ref{sec:flux-connectivity}, and our analysis focuses on finite-time $\ell^1$ error bounds via flux-based estimates and semigroup contraction. We do not yet track how spectral quantities (such as the spectral gap or mixing time) change as states are added and removed. A natural next step would be to make this connection rigorous by combining conductance bounds with ergodicity-coefficient techniques \cite{Rhodius2000Ergodicity, RHODIUS1997141} and perturbation bounds for Markov semigroups \cite{pazy_1983}, using flux information to control how much each pruning event can contract the Dobrushin coefficient. This would complement the finite-time, flux-based error analysis developed here by providing spectral stability guarantees for the truncated generators.

From an implementation perspective, several steps of the algorithm can be parallelized. Boundary expansion factorizes over reactions, matrix construction over columns, and Krylov iterations over sparse matrix–vector products. Distributed-memory implementations that partition the state space and communicate flux information along partition boundaries could extend the reach of the method to larger problems without changing the underlying analysis.

Finally, we have focused on time-homogeneous generators. Time-dependent propensities $\alpha_k(x,t)$, arising for example in driven systems or externally controlled networks, fit naturally into the flux-based framework: both the exit rates $w(x,t)$ and the total flux $\Phi_{\mathrm{total}}(t)$ become explicitly time dependent. The same error bounds apply on each step, provided the generator is frozen over that step, and the adaptive step-size rule still responds directly to the changing activity level. A more detailed analysis of time-dependent generators, and of how closely truncated semigroups track the long-time behavior of the full system, would complement the finite-time flux-based error bounds derived here.

\section{Conclusion}
\label{sec:conclusions}

We presented an adaptive FSP method that addresses fundamental challenges in stiff chemical kinetics through flux-based pruning to preserve network connectivity and flux-adaptive time stepping to handle varying activity levels across disparate timescales. The flux-based mechanism identifies and retains low-probability bottleneck states that carry significant flux, preventing catastrophic failures where probability-only truncation severs essential reaction pathways. We derived global error bounds showing that errors accumulate additively due to contraction properties of stochastic generators. Numerical experiments on stiff and metastable benchmark systems demonstrate that the method maintains accuracy with relative errors $\sim$0.01--0.02\% while using state spaces orders of magnitude smaller than exact approaches. A head-to-head comparison with Krylov-FSP-SSA reveals three qualitatively distinct failure modes---state-space explosion, wavefront truncation, and stiffness-induced $\tau$-lock---and confirms that flux-aware pruning and flux-adaptive time stepping address all three. The method successfully handles problems where probability-only pruning fails by removing connector states, though threshold selection requires problem-specific tuning based on system stiffness and timescale separation.

\section{Code Availability}
The implementation is publicly available as a Julia package at \url{https://github.com/AdityaDendukuri/DiscStochSim.jl}.

\section{Acknowledgments} We thank the anonymous reviewers for their valuable feedback. The suggestion to examine the bottleneck state example was particularly impactful in clarifying the core motivation for our flux-based approach. The first author also acknowledges Ricky Lee (UC Santa Barbara) for helpful discussions on the mathematical derivations.

\bibliographystyle{siamplain}
\bibliography{references}
\end{document}